\documentclass[english, 11pt]{article}
\usepackage[round]{natbib}
\usepackage{color}
\usepackage{authblk}
\usepackage[colorlinks=true, linkcolor=teal, urlcolor=purple, citecolor=purple]{hyperref}
\usepackage[hyperref]{xcolor}
\usepackage{amsmath}
\usepackage{amsfonts}
\usepackage[plain,noend]{algorithm2e}
\usepackage{amsthm}
\usepackage{graphicx}
\usepackage{caption}
\usepackage{subcaption}
\usepackage{booktabs}
\usepackage[normalem]{ulem}
\usepackage{tikz}
\usepackage{lineno}
\usepackage{multirow}
\usetikzlibrary{decorations.pathreplacing,calligraphy}
\usepackage{cleveref} 
\newtheorem{theorem}{Theorem}
\newtheorem{lemma}[theorem]{Lemma}
\newtheorem{corollary}[theorem]{Corollary}
\theoremstyle{definition}
\newtheorem{definition}[theorem]{Definition} 
\newtheorem{example}{Example} 
\newtheorem{remark}{Remark} 
 
\usepackage[a4paper]{geometry}
\newcommand{\tbl}[1]{\caption{#1}}
\newenvironment{tabnote}{
  \par\vskip0.5em
  \small\noindent\textit{Note:}\ }{}

\newtheorem{algo}{Algorithm}

\title{SOMA: A Novel Sampler for Bayesian \\ Inference from Privatized Data}

\author{Yifei Xiong$^1$ , Nianqiao Phyllis Ju$^2$\thanks{Corresponding author: nianqiao.ju@dartmouth.edu. Portions of this work were conducted while this author was at Purdue University, Department of Statistics.}}
\date{%
    $^1$Department of Statistics, Purdue University\\
    $^2$Department of Mathematics, Dartmouth College
}

\begin{document}
\maketitle

\begin{abstract}
Making valid statistical inferences from privatized data is a key challenge in modern analysis. In Bayesian settings, data augmentation MCMC (DAMCMC) methods impute unobserved confidential data given noisy privatized summaries, enabling principled uncertainty quantification. However, standard DAMCMC often suffers from slow mixing due to component-wise Metropolis-within-Gibbs updates. We propose the Single-Offer-Multiple-Attempts (SOMA) sampler. This novel algorithm improves acceptance rates by generating a single proposal and simultaneously evaluating its suitability to replace all components. By sharing proposals across components, SOMA rejects fewer proposal points. We prove lower bounds on SOMA’s acceptance probability and establish convergence rates in the two-component case. Experiments on synthetic and real census data with linear regression and other models confirm SOMA’s efficiency gains.
\end{abstract}

\tableofcontents

\section{INTRODUCTION}\label{sec:intro}

Imputing data from noisy summary statistics is a foundational problem with broad applications, especially in privacy-preserving data analysis.
In these applications, the data analysts cannot access the raw data directly, and only the noise-corrupted results from certain summary statistics queries are shared with them.
By sampling from the conditional distribution of the missing raw data given observed noisy summary statistics, data imputation enables downstream tasks that require individual-level data.
As the target distributions are multivariate and their densities can be evaluated only up to an unknown normalizing constant, 
data imputation typically requires Markov Chain Monte Carlo samplers.

\citet{ju2022data} proposes a data augmentation~\citep{van2001art} strategy to make inference from privatized data. The data augmentation MCMC (DAMCMC) strategy has since also been adopted by \citep{guo2024differentially,awan2024statistical} among others. 
DAMCMC consists of an inference step and an imputation step.
This gives a chain that approximates the joint distribution of confidential data and parameters given observed noisy summary statistics.
For the imputation step,
\citet{ju2022data} propose using the independent Metropolis-Hastings (IMH)~\citep{tierney1994markov,liu1996metropolized,lee2018optimal} sampler for each conditional update within a Gibbs sampler.
We refer to this strategy as independent-Metropolis-within-Gibbs (IMwG).
They show that the acceptance probability of IMwG proposed states are lower bounded, and connected to the magnitude of privacy noise.
This makes IMwG suitable for the problem of interest.
However, there remains an opportunity to design new sampling algorithms with a provably faster convergence than the IMwG method.

Our contributions are three folded. 
\begin{itemize}
    \item First, we proposes a novel Single-Offer-Multiple-Attempts (SOMA) sampler that exploits a permutation-invariant property of the data imputation distribution. To the best of our knowledge, no existing method has exploited this property.
    \item 
    Second, we investigate the theoretical properties of SOMA and IMwG. We show that the acceptance probability of SOMA is larger than that of IMwG. For the case of $n = 2$ components, we establish convergence rates of SOMA and IMwG by constructing a Markovian coupling. We also prove that SOMA becomes rejection-free as $n \to \infty.$
    \item Finally, we apply SOMA to the problems of Bayesian linear regression given privatized data and Bayesian inference from privatized compositional data on American Time Use Survey (ATUS) data.
\end{itemize}

The remainder of this paper is structured as follows: Section~\ref{sec:motivations} defines the data imputation problem, establishes the permutation-invariant property of the target distribution, and reviews the IMwG method and its limitations.
We present the SOMA algorithm in Section~\ref{sec:soma} and analyze its acceptance probability and convergence rate in Section~\ref{sec:theory}. Simulation and real data experiments are presented in Section~\ref{sec:exp} \&~\ref{sec:realdata}. The appendix contains proofs and additional experiments.

\section{PROBLEM FORMULATION AND CHALLENGES}\label{sec:motivations}
The primary motivation of SOMA is to make inferences from observed data that have been corrupted to satisfy privacy guarantees. 

\subsection{Differential Privacy and Data Release}\label{sec:privacy}
Consider a confidential dataset $X = (X_1, \ldots, X_n)^\top$ consisting of $n$ records, and each $X_i \in \mathbb{X}.$
When $X$ contains sensitive information, only sanitized results $S_{\mathrm{dp}} \in \mathbb{S}$ are released to the data users for analysis.
Note that $\mathbb{S}$ need not have the same dimension with $\mathbb{X}$ or $\mathbb{X}^n$. When $S_{\mathrm{dp}}$ is a sanitized database, we have $S_{\mathrm{dp}} \in \mathbb{X}^k$, where $k$ is the number of sanitized samples. When $S_{\mathrm{dp}}$ is some privatized query result, $S_{\mathrm{dp}}$ has the same dimension as the query. The methodology we develop is applicable to both instances.

A data release mechanism $\eta( \cdot \mid X)$ is a conditional distribution of $S_{\mathrm{dp}}$ given $X$.
This paper focuses on mechanisms that satisfy the definition of differential privacy (DP) from~\citet{dwork2006differential}.

\begin{definition}[Distance between datasets]
Let $x,x^{\prime}$ be two datasets of the same size.
Define $d(x,x^\prime)$ as the minimum number of components in which $x$ and $x^\prime$ differ, up to permutation of indices, i.e., 
\begin{equation*}
d(x, x^\prime) = \min_{\sigma \in S_n}\|x_{\sigma} - x'\|_{0},
\end{equation*}
where we use $x_{\sigma}$ to represent the vector obtained from permuting indices from $x$. We call $x,x^{\prime}$ neighboring datasets if $d(x,x^{\prime}) = 1.$
\end{definition}

\begin{definition}
A mechanism $S_{\mathrm{dp}} \sim \eta(\cdot \mid X)$ satisfies $\epsilon$-DP, if, for all pairs of neighboring databases $(x, x^\prime)$, we have
\begin{equation}\label{eqn:epsilon-dp}
\sup_{B \in \mathcal{B}} \frac{\eta(B \mid x)}{\eta(B \mid x')} \le \exp(\epsilon),  
\end{equation}
where $\mathcal{B}$ is the measurable ball on $\mathbb{S}$
where the parameter $\epsilon>0$ is called the \textit{privacy loss budget}. The ratio is defined to be 1 whenever the numerator and denominator are both 0.
\end{definition}

The Laplace mechanism is a common technique for adding noise to achieve privacy. Given a query function $s: \mathbb{X}^n \to \mathbb{S}$, and query result $s(X)$, we have $S_{\mathrm{dp}} = s(X) + \nu$ where components in $\nu$ are independent and identical samples from a zero-centered Laplace distribution. The variance is determined by the sensitivity of $s$ as well as the privacy loss budget $\epsilon$. We give some examples of privatized query results here. These mechanisms will be covered by our experiments in Sections~\ref{sec:exp}, \ref{sec:realdata}, and \ref{sec:exp_hist}.

\begin{example}[Perturbed histogram]\label{example:perturbedhistogram}

Let the confidential data domain be $\mathbb{X} = [0,1]^r$. 
We create $m$ bins in $\mathbb{X}$, and denote the $j$-th bin as $B_j$. Let $S_j$ be the number of $X_i$'s following in the $j$-th bin. Here $(S_1,\ldots,S_m)$ is a query result from $X$. 
Let $S_{\mathrm{dp},j} = S_j + \nu_j$ where $\nu_j$ are independent draws from 
$\mathrm{Laplace}(0, 2/\epsilon).$
The perturbed histogram mechanism~\citep{dwork2006differential} is to release $S_{\mathrm{dp}}$. One can also generate synthetic data from perturbed histograms~\citep{wasserman2010statistical}. Let $\tilde{S}_j = \max(0,S_{\mathrm{dp},j})$, $q_j = \tilde{S}_j / \left(\sum_{j=1}^k \tilde{S}_j\right)$. Sample $Z_i$ independently according to (unormalized) density $\sum_{j=1}^m q_j \mathbb{I}(z \in B_j)$ for $i = 1,\ldots, k$, gives a sanitized database $Z = (Z_1,\ldots,Z_k)$ of size $k$.
\end{example}

\begin{example}[Data release for privatized linear regression]\label{example:linearregression}
Consider a regression setting with predictors $x_i \in \mathbb{R}^p$ and a response $y_i \in \mathbb{R}$. The confidential data is a concatenated matrix $Z$ with rows $(x_i^\top, y_i)$.
To ensure finite sensitivity, each entry in $Z$ is first clamped to a pre-defined range. The summary statistics are the unique elements of the Gram matrix $Z^\top Z$.
A privatized version is released by adding iid Laplace noise, scaled to the sensitivity of these statistics with each unique element~\citep{bernstein2019differentially}.
\end{example}

\begin{example}[Private release of compositional data]\label{example:compositionaldata}
Let $\mathbb{X} = \{0 \le x_j \le 1, \sum_{i=1}^p x_j = 1\}$. 
The summary statistic is constructed from transformed data. 
First, each record $X_i$ is processed by element-wise clamping to an interval $[a, 1]$, followed by applying the logarithm. 
The sum of these transformed vectors, $S_{\log} = \sum_{i=1}^n \log([X_i]_a^1)$.
A privatized version is then released as $S_{\mathrm{dp}} = S_{\log} + \nu$, where the components of the noise vector $\nu$ are iid draws from a $\mathrm{Laplace}(0, -p\log(a)/\epsilon)$ distribution~\citep{guo2024differentially}.
\end{example}

\subsection{Inference from Privatized Data and Data Augmentation}\label{sec:imputation}
Our goal is to make inferences about the data-generating process $X_i \overset{iid}{\sim}f_{\theta}(\cdot) = f(\cdot \mid \theta)$ based on privatized data $S_{\mathrm{dp}}$.
Sometimes, a data analyst can only access some query result $s(x)$ about the confidential dataset, and the data maintainer might only share a privatized version of it, resulting in $s_{\mathrm{dp}} = s(x) + z$, where $z$ is some additive noise independent of $s(x).$
The observation density $\eta(s_{\mathrm{dp}} \mid s(x))$ depends on $s_{\mathrm{dp}} - s(x)$ and known privacy parameters.

In this setting, we need to make inferences about $\theta$ based on $s_{\mathrm{dp}}$, without direct access to $x$.
This article focuses on Bayesian inference, where the data analyst describes their prior belief about $\theta$ with some prior distribution $p(\theta).$
\citet{ju2022data} proposes a data augmentation strategy~\citep{van2001art} using a Markov Chain Monte Carlo algorithm to approximate the joint posterior distribution 
\begin{equation}\label{eqn:joint-posterior}
    \pi(\theta, x \mid s_{\mathrm{dp}}) \propto p(\theta)\cdot \prod_{i=1}^n f(x_i \mid \theta) \cdot  \eta( s_{\mathrm{dp}} \mid s(x)).
\end{equation}
From \eqref{eqn:joint-posterior}, the marginal distribution of $\theta$ is $\pi(\theta \mid s_{\mathrm{dp}})$. 
The algorithm alternates between two steps:
\begin{itemize}
\item Inference step: (approximate) sampling from the conditional distribution $\pi(\theta \mid x, s_{\mathrm{dp}});$
\item Imputation step: (approximate) sampling from the conditional distribution 
\begin{equation}\label{eqn:data-imputation-target}
\pi(x \mid \theta, s_{\mathrm{dp}}) \propto \eta(s_{\mathrm{dp}} \mid s(x)) \prod_{i=1}^n f(x_i \mid \theta)
\end{equation}
\end{itemize}
When focusing on the imputation step, $\theta$ and $s_{\mathrm{dp}}$ are fixed, and 
we can write $\pi(x \mid \theta, s_{\mathrm{dp}})$ as $\pi(x)$ to simplify the notation.
This work designs a sampler for the imputation step by exploiting the permutation invariant property of \eqref{eqn:data-imputation-target}.

\begin{definition}[Permutation invariance]\label{assump:exchange}
A distribution $\pi(x) = \pi(x_1,\ldots,x_n)$ is permutation invariant if, 
for any permutation $\sigma$ of the group $\{1,2,\ldots,n\}$, we have 
$
\pi(x_{\sigma(1)},x_{\sigma(2)}, \ldots,x_{\sigma(n)}) = \pi(x_1,x_2,\ldots,x_n).
$
In other words, $(X_1,\ldots,X_n) \sim \pi$ is an exchangeable sequence of random variables.
\end{definition}

\begin{lemma}
Whenever the query $s(x)$ is permutation invariant, the conditional distribution \eqref{eqn:data-imputation-target} is also permutation invariant.
\end{lemma}
\begin{proof} 
\vspace{-3pt}
Since $s(x)$ and $\prod_{i=1}^{n} f(x_i \mid \theta)$ are both permutation invariant, their product \eqref{eqn:data-imputation-target} is also permutation invariant.
\vspace{-3pt}
\end{proof}

Many queries like the mean, median, empirical distribution, or covariance are invariant to the ordering of the components. 
In particular, \Cref{example:perturbedhistogram,example:compositionaldata,example:linearregression} all lead to permutation-invariant imputation steps.
As a result, the imputation step is to sample from a permutation-invariant distribution. Previous work~\citep{ju2022data,awan2024statistical,xiong2025simulation} have not leveraged this property of $\pi(x \mid s_{\mathrm{dp}}, \theta).$

The methodology we propose is applicable to any multivariate distribution that is permutation invariant, which arises in various other problems in statistics and machine learning.

\subsection{Gibbs, Metropolis, and Metropolis-within-Gibbs}\label{sec:mwg}
When direct sampling from the joint distribution $\pi(x)$ is difficult, a Markov chain Monte Carlo strategy is the Gibbs sampler~\citep{geman1984stochastic,gelfand1990sampling}.
Gibbs sampling is ubiquitous in Bayesian modeling, and it is particularly useful when one can sample from the full conditional distributions 
\begin{equation}\label{eqn:fullconditional}
    \pi_i(x_i \mid x_{-i}) = \frac{\pi(x_i,x_{-i})}{\int \pi(x_i,x_{-i})\ \mathrm{d} x_i} \propto \pi(x_i, x_{-i}).
\end{equation}
Here $x_{-i}$ includes all components of $x$ except the $i$-th component, i.e., $x_{-i} = (x_1,\ldots,x_{i-1},x_{i+1},\ldots,x_n).$
Note that permutation invariance allows us to write the conditional as $\pi(x_i \mid x_{-i})$ instead of $\pi_i(x_i \mid x_{-i}).$

Gibbs samplers entails iterative updates according to $x_i \sim \pi(\cdot \mid x_{-i}).$
There are two main update schedules for the Gibbs sampler~\citep{andrieu2016random}: systematic-scan (SysScan, also called deterministic-scan) and random-scan (RanScan). The former method sequentially updates all components once at a time with a fixed schedule $i = 1,2,\ldots,n$, and the latter randomly selects a component $i$ to perform the conditional update in each iteration.

In many practical scenarios and especially for data imputation tasks
the full conditional distribution $\pi(x_i \mid x_{-i})$ is intractable and hence direct sampling is difficult.
In that case, we must resort to approximated Gibbs samplers~\citep{qin2025spectral} such as the Metropolis-within-Gibbs sampler~\citep{gilks1995adaptive}, which uses a Metropolis-Hastings kernel to approximate each conditional distribution.

The inherent symmetry from permutation invariance of \eqref{eqn:data-imputation-target} suggests a certain homogeneity among the full conditionals \eqref{eqn:fullconditional}. 
\citet{ju2022data} notices that it is reasonable to use the same Metropolis kernel in each Gibbs step. 
In fact, they propose to use the independent Metropolis-Hastings~\citep{tierney1994markov,liu1996metropolized,lee2018optimal,wang2022exact} sampler to approximate \eqref{eqn:fullconditional}, making the entire algorithm IMH within Gibbs (IMwG).
In general, IMH uses proposal distributions that are independent of the current state. 
In the context of IMwG, when updating the $i$-th component, IMH proposes a new state $y \sim q(\cdot)$ from $\mathbb{X}$, independent of the current state $x_i$. 
Replacing $x_i$ by $y$ with probability 
\begin{align}
\alpha^{\mathrm{IMwG}}_i(y, x_i \mid x_{-i}) &=   \min\left(1, \frac{\pi(y \mid x_{-i})q(x_i)}{\pi(x_i\mid x_{-i})q(y)}\right)= \min\left(1, \frac{\pi(y, x_{-i})q(x_i)}{\pi(x_i, x_{-i})q(y)}\right) \label{eqn:imh-acceptance}
\end{align}
leaves the (conditional) target distribution $\pi(x_i \mid x_{-i})$ invariant. 
If accepted, then the new state $x' = [x_{-i},y]$ is the result of replacing the $i$-th component of $x$ by $y$ and keeping all other components.

With full conditional \eqref{eqn:fullconditional} and proposing from the model $q(\cdot) = f(\cdot \mid \theta)$ simplifies \eqref{eqn:imh-acceptance} to 
\begin{equation}
    \label{eqn:imh-acceptance-simple}
    \min\left(1, \frac{\eta(s_{\mathrm{dp}} \mid s([x_{-i}, y]))}{\eta(s_{\mathrm{dp}} \mid s(x))}\right).
\end{equation} This computation can be completed in $\mathcal{O}(1)$ time when the privacy mechanism satisfies a `recording additivity' property~\citep[Assumption 2]{ju2022data}.

We will compare our proposed method with IMwG, under both RanScan and SysScan schedules. 
We refer to these two baselines as Ran-IMwG and Sys-IMwG, respectively.

\section{A SINGLE-OFFER-MULTIPLE-ATTEMPTS SAMPLER}\label{sec:soma}
We describe our novel SOMA sampler in this section. 
SOMA is motivated by the following observation about IMwG. When a proposed value $y$ to replace $x_i$ is rejected, it can still be a favorable candidate to swap out a different component $x_j$. 
In this sense, plain IMwG is not using the proposed samples efficiently.

Unlike the IMwG strategy of comparing $y \in \mathbb{X}$ with one component $x_i$ at a time, 
SOMA simultaneously considers a single proposed state $y$ as a candidate to replace multiple components. 
With a sample $y$, SOMA first randomly chooses a component $I$ from $1, 2 \ldots,n$ to replace. 
The selection probability is proportional to weights $w_i$ given by
\begin{equation}\label{eqn:wi}
w_i(y,x) = \eta(s_\mathrm{dp} \mid s([x_{-i},y])),\quad \text{for } i = 1,\ldots, n.
\end{equation}
The probability of choosing the $i$-th component is then $\mathbb{P}(I = i \mid y, x) = w_i(y,x) / \left(\sum_{i=1}^n w_i(y,x)\right)$. 
Using weights \eqref{eqn:wi} means that SOMA favors swapping out components that are less likely to have produced $s_\mathrm{dp}$. Notice that $w_i$ coincides with the denominator in \eqref{eqn:imh-acceptance-simple}.

In the second stage, given $I = i$, SOMA attempts to swap out $x_i$ and replace it by $y$ with probability 
\begin{equation}\label{eqn:acceptprob}
\alpha_i^{\mathrm{SOMA}}(y,x)=\min\left\{1, \frac{W(y, x)}{W(y, x) + w_0(y, x) - w_i(y, x)}\right\},
\end{equation}
where
\begin{equation}\label{eqn:w0}
w_0(y, x)=\eta(s_\mathrm{dp}\mid s(x))
\end{equation}
is defined analogously to \eqref{eqn:wi}.

\begin{algo}\label{alg:soma}
SOMA sampler for $\eqref{eqn:data-imputation-target}$.
\begin{tabbing}
\textbf{Stage 1}:\\
\quad Propose $y \sim q(\cdot) = f(\cdot \mid \theta)$, \\
\quad Calculate weights $w_i=\eta(s_\mathrm{dp} \mid s([x_{-i},y]))$, $w_0=\eta(s_\mathrm{dp}\mid s(x))$ and $W=\sum_{i=1}^n w_i$, \\
\quad Select random index $I$ from $\{1,2,\ldots,n\}$ with probability $\mathbb{P}(I=i) \propto w_i$, \\
\textbf{Stage 2}:  With probability $\alpha = \min \left\{1, \frac{W}{W+w_0-w_I}\right\}$, replace $x_I$ by $y$.
\end{tabbing}
\end{algo}

While IMwG exploits the symmetry of permutation invariant distributions by using the same independent proposal for each component, SOMA has pushed this idea further.
SOMA recognizes that an unfavorable proposal for $x_j$ might be a favorable proposal for some other $x_i$, and it seeks to find such an index $i$ by comparing all the weights given in \eqref{eqn:wi}.

The advantage of SOMA lies in its efficient use of the `single offer' from the proposal distribution $q$. 
Stage 1 of SOMA chooses a favorable component $i$ among the multiple choices to swap out $x_i$ by $y$.
Then Stage 2 compares the current state $x$ with the proposed state $x' = [x_{-i},y]$. 
The acceptance probability \eqref{eqn:acceptprob} accounts for the fact that $x'$ is selected from multiple attempts and is not a direct comparison like \eqref{eqn:imh-acceptance}.
We summarize SOMA in Algorithm~\ref{alg:soma} and show that it is reversible with respect to the target $\pi$ in Theorem~\ref{thm:rev}. The proof is deferred to Section~\ref{sec:proofs_of_section_3}.

\begin{theorem}[Reversibility]
\label{thm:rev}
Assume $\pi$ is permutation invariant. Stage 1 of Algorithm~\ref{alg:soma} generates a proposal kernel, which we denote by $K^{\mathrm{SOMA}}(x' \mid x)$. 
\begin{itemize}
    \item We have $K^{\mathrm{SOMA}}(x^\prime\mid x)=0$ for all $d(x,x^\prime)>1$. 
    \item 
    When $d(x,x^\prime) = 1$, suppose they only differ on the $i$-th entry, i.e., $x_i^\prime = y$ and $x_j^\prime = x_j$ for all $j\neq i$, then  $K^{\mathrm{SOMA}}(x^\prime\mid x)=q(y){w_i(y, x)}/{W(y, x)}.$
    \item Associating proposal kernel $K^{\mathrm{SOMA}}$ with acceptance probability \eqref{eqn:acceptprob} leads to a Markov transition kernel $P_{\mathrm{SOMA}}$ reversible with respect to $\pi$.
\end{itemize}
\end{theorem}

\begin{remark}
Theorem~\ref{thm:rev} shows that SOMA transforms the independent proposal $q$ into an informed proposal kernel $K^{\mathrm{SOMA}}$ that is dependent on the current state. This distinguishes SOMA from IMH. 
\end{remark}

\begin{corollary}[Sufficient conditions for reversibility]\label{coro:sufficientconditions} The following conditions are sufficient for reversibility.
\begin{itemize}
    \item $\pi$ is a permutation invariant distribution, and 
    \item The weights are `proper' in the sense
\begin{align*}
w_0(y,x) & = c \cdot \frac{\pi(x)}{\prod_{j=1}^n q(x_j)},\ \text{and } \\
w_i(y,x) & = c\cdot \frac{\pi([x_{-i},y])}{q(y)\prod_{j=1,j\neq i}^n q(x_j)},\ \forall i=1,\ldots,n
\end{align*}
for some fixed constant $c >0.$
\end{itemize}
\end{corollary}
This corollary shows that SOMA can be applied for any permutation invariant target. Applying SOMA for data imputation from privatized summary statistics is a special and important application. 

\begin{remark}\label{remark:soma-vs-imwgibbs}
We should also distinguish the SOMA sampler from the Ran-IMwG and Sys-IMwG.
To begin with, SOMA is not a Gibbs sampler - it is a Metropolis-Hastings type algorithm, consisting of a proposal step and an accept/reject step. Instead, Stage 1 of SOMA constructs a Metropolis kernel $K^{\mathrm{SOMA}}$ by converting the IMH proposal into a dependent one.
Second, although SOMA and Ran-IMwG randomly pick an index $i$ to perform a swap move, Ran-IMwG uses fixed weights independent of the current state, while SOMA uses weights~\eqref{eqn:wi} that adapt to the current state.
We can see SOMA as a multiproposal MCMC~\citep{pozza2025fundamental,lin2025quantum} method, and Ran-IMwG is the `single proposal' counterpart that SOMA should compare with. 
More discussions are in Section~\ref{sec:multiproposal-1}.
\end{remark}

\begin{remark}
Multiple-try Metropolis (MTM)~\citep{liu2000multiple} also entails multiple evaluations in one Metropolis iteration. MTM would require multiple proposed points, independent or correlated~\citep{craiu2007acceleration}. The Multiple-try Metropolis Independent sampler (MTM-IS)~\citep{yang2023convergence} bears resemblance to SOMA in using the IMH kernel. 
MTM-IS generates multiple samples $\{y_1,\ldots,y_n\}$ from an independent proposal $q$, selects one index from them, and generates `balancing trials' to satisfy reversibility. On the other hand, SOMA only requires one sample $y$ from $q$. The multiple attempts for SOMA are created by comparing $y$ with each component in the current state $x$. SOMA would require much less computation and memory cost than MTM-IS.
\end{remark}

\begin{remark}Per-iteration cost of SOMA is $\mathcal{O}(n).$
It is natural to modify SOMA to simultaneously consider a subset of $m < n$ components, and thus reducing the cost. Let $M$ be the index set. In Stage 1, one can compute only the weights $w_0$ and $w_i$ for $ i \in M$. In Stage 2, the acceptance probability would be 
$\min\left(1, \frac{\sum_{i \in M} w_i}{w_0 + \sum_{j\not= I, j \in M} w_j}\right).$
Our theoretical analysis considers the standard case, where all $n$ components are considered. 
The reversibility still holds for this $m$-component variant.
See Appendix~\ref{app:soma_on_subset} for details.
\end{remark}

\section{CONVERGENCE OF SOMA}\label{sec:theory}
We study the convergence characteristics of SOMA in this section, by focusing on its convergence rate and comparison with Ran-IMwG and Sys-IMwG.
\begin{definition}[Geometric ergodicity]
A Markov chain with transition kernel $P(\cdot \mid \cdot)$ invariant to $\pi$ is geometric ergodic in total variation distance if the following inequality holds for some function $C$, constant $r \in (0,1)$ and $\pi$-almost all $x$:
\begin{equation}\label{eq:def_convergence_rate}
\|P^t (\cdot \mid x) - \pi(\cdot) \|_{\mathrm{TV}} \le C(x) r^t,
\end{equation}
where the left-hand side is the total variation distance between the target distribution and the distribution of the Markov chain at iteration $t$ when it started at $x$.
\end{definition}
The constant $r\in (0, 1)$ gives the convergence rate: the smaller $r$ is, the faster the convergence.

In Section~\ref{sec:acc_bound}, we first examine the acceptance probabilities of the three algorithms, which play a crucial role in determining convergence efficiency. 
In Section~\ref{sec:convergence}, we derive an upper bound on the convergence rate. In Section~\ref{sec:cost}, we provide implementation strategies for SOMA and compare its computational cost with other methods.

\subsection{Acceptance Rate}\label{sec:acc_bound}
For accept-reject-based Markov chains, a small acceptance probability is a clear signal for slow convergence~\citep{brown2022lower}.
For component-wise Markov chains, slow convergence of the Metropolis kernels makes the overall Metropolis within Gibbs sampler mix slowly~\citep{qin2025spectral}.

First, we analyze the coordinate-wise acceptance probability of the three algorithms. This is the probability of accepting a proposed move conditioning on the proposed point $y$ as well as the chosen coordinate $i$ for a swap move.

\begin{theorem}[Coordinate-wise acceptance]\label{thm:acc_compare}
Let $\alpha_i(y,x)$ denote the acceptance probability of replacing the $i$-th component of $x \in \mathbb{X}^n$ with $y \in \mathbb{X}$, given that index $i$ has been selected for the update.
For the SOMA sampler and the IMwG samplers, the corresponding Metropolis acceptance probability satisfies
\begin{equation}
\alpha^{\mathrm{IMwG}}_i (y,x) = \min\left\{1, \frac{w_i(y,x)}{w_0(y,x)}\right\}\le \alpha_i^{\mathrm{SOMA}}(y,x).
\end{equation}
\end{theorem}
See complete proof in Section~\ref{sec:proofs_of_section_41}. 

Theorem~\ref{thm:acc_compare} indicates that, if all three methods choose to update the same component $i$, SOMA always has an acceptance probability at least as high as the other IMwG methods.
\begin{theorem}\label{thm:peskun} Let $P_{\mathrm{SOMA}}$ and $P_{\mathrm{RAN}}$ be the Markov transition kernels associated with SOMA and Ran-IMwG. We must have, for $\pi$-almost all $x,x'$ with $d(x,x') = 1$, 
\begin{equation}\label{eqn:peskun-x}
P_{\mathrm{SOMA}}(x' \mid x) \le n P_{\mathrm{RAN}}(x' \mid x).
\end{equation}
\end{theorem}
\Cref{thm:peskun} indicates an ordering between the asymptotic variances of the samplers. We will discuss this more in Section~\ref{sec:cost}. \Cref{thm:peskun} can also be seen as a consequence of \citet[Theorem 1]{pozza2025fundamental}, and we include this alternative proof in Section~\ref{sec:multiproposal-1}.

Next, we analyze the overall acceptance probabilities. 
Let $A(x)$ be the probability of accepting a move from $x$. 
For SOMA and Ran-IMwG, the coordinate $i$ is chosen randomly, and we have
\begin{align}
A^{\mathrm{SOMA}}(x) &= \int_{\mathbb{X}} q(y) \sum_{i=1}^{n}  \left(\frac{w_i(y,x)}{W(y,x)} \right) \cdot \alpha_i^{\mathrm{SOMA}}(y,x)\mathrm{d}y, \label{eqn:soma-accept} \\
A^{\mathrm{RAN}}(x) &= \int_{\mathbb{X}} q(y) \sum_{i=1}^{n} \left( \frac{1}{n}\right) \cdot \alpha_i^{\mathrm{IMwG}}(y,x) \mathrm{d}y. \label{eqn:ran-accept}
\end{align}
\eqref{eqn:soma-accept} and \eqref{eqn:ran-accept} reiterate Remark~\ref{remark:soma-vs-imwgibbs} that Ran-IMwG uses fixed weights while SOMA adaptively decides which component to update.
Sys-IMwG has a predetermined update order, and we write the probability of accepting a move to change the $i$-th component as
\begin{equation}\label{eqn:sys-accept}
A_i^{\mathrm{SYS}}(x) = \int_{\mathbb{X}} q(y) \alpha^{\mathrm{IMwG}}_i(y,x) \mathrm{d}y. 
\end{equation}

We establish in \Cref{thm:acc_soma_vs_ran} that, from each current state $x$, SOMA has a higher acceptance probability than RanScan.

\begin{theorem}[Comparison of acceptance probabilities]\label{thm:acc_soma_vs_ran}
For any $x\in\mathbb{X}^n$, we have:
\begin{equation}
A^{\mathrm{SOMA}}(x) \ge A^{\mathrm{RAN}}(x).
\end{equation}
\end{theorem}

\begin{corollary}\label{corollary:tv-ran-soma}
Let $P_{\mathrm{RAN}}$ be the Markov transition kernels associated with Ran-IMwG, then for all $t$,
the total variation distance in \eqref{eq:def_convergence_rate} satisfied
\begin{align*}
\|P^t_{\mathrm{RAN}} (\cdot \mid x) - \pi(\cdot) \|_{\mathrm{TV}} & \ge \left(1-A^{\mathrm{RAN}}(x)\right)^t \ge \left(1-A^{\mathrm{SOMA}}(x)\right)^t.
\end{align*}
\end{corollary}
Corollary~\ref{corollary:tv-ran-soma} follows \Cref{thm:acc_soma_vs_ran} and \citet[Theorem 2]{brown2022lower}. It suggests that if $A^{\mathrm{SOMA}}(x) \to 0$ (e.g., as $n \to \infty$), then the lower bound on the convergence rates for both SOMA and Ran-IMwG approaches 1, indicating slow convergence. However, if $A^{\mathrm{SOMA}}(x)$ remains bounded away from zero, Ran-IMwG may still suffer from poor convergence, while SOMA can potentially avoid this issue by maintaining a higher acceptance probability.

We now give lower bounds for the overall acceptance probability from the current state $x$ under $\epsilon$-DP.

\begin{theorem}[Local acceptance probability]\label{thm:acc_bound}
Let $A^{\mathrm{SOMA}}(x)$, $A^{\mathrm{RAN}}(x)$ and $A_i^{\mathrm{SYS}}(x)$ be defined as in \eqref{eqn:soma-accept}, \eqref{eqn:ran-accept} and \eqref{eqn:sys-accept} respectively. These are the overall acceptance probabilities from a current state $x$. Under $\epsilon$-DP, we have
\begin{align}
&A^{\mathrm{SOMA}}(x) \ge \frac{n}{n+e^\epsilon-1}, \ 
A^{\mathrm{RAN}}(x) \ge \frac{1}{e^\epsilon}, \ \text{and}\ A_i^{\mathrm{SYS}}(x) \ge \frac{1}{e^\epsilon} \ \forall i. \label{eqn:accept-lower-bound}
\end{align}
\end{theorem}
The lower bounds on acceptance probability place the three algorithms into a `no clear signal of slow convergence' scenario. 
We also highlight that when $n$ is large, the acceptance probability of SOMA approaches 1.
As a result, SOMA becomes almost rejection-free in high dimensions.

In fact, our experiments in \ref{sec:exp_linear} for Bayesian linear regression already reports SOMA acceptance rate between 92\% and 99.5\% for sample size $n = 10$, depending on the privacy loss budget $\epsilon.$

\subsection{Upper Bound of the Convergence Rate}\label{sec:convergence}

In this section, we use coupling-based techniques to analyze the convergence rate of the algorithms. By constructing a faithful coupling (see Appendix~\ref{sec:coupling} for technical details and the coupling kernel), we establish the following upper bounds on the convergence rates for the $n=2$ case:

\begin{theorem}[Upper bound on convergence rate]\label{thm:convergence_rate}
For the case $n=2$, under $\epsilon$-DP, the convergence rate of the SOMA algorithm is upper bounded by
\begin{equation}\label{eq:rate_soma}
r_{\mathrm{SOMA}} \le \frac{e^{2\epsilon}+3e^\epsilon-2 + \sqrt{e^{4\epsilon} + 2e^{3\epsilon} + 9e^{2\epsilon} - 8e^\epsilon}}{2(1+e^\epsilon)^2}.
\end{equation}
For comparison, the upper bounds on the convergence rates for the Ran-IMwG and Sys-IMwG samplers are given by
\begin{align}\label{eq:rate_ran_sys}
r_{\mathrm{RAN}} &\le \frac{3e^\epsilon-2+\sqrt{e^{2\epsilon}+4e^\epsilon-4}}{4e^\epsilon},\notag \\
r_{\mathrm{SYS}} &\le \frac{\sqrt{(e^\epsilon+1)(e^\epsilon-1)}}{e^\epsilon},
\end{align}
respectively.
\end{theorem}
The proof constructs a transition matrix of the between-chain distance $d(x^{(t)}, \tilde{x}^{(t)})$. The hitting time to the absorbing state $d = 0$ is the coupling time, which yields an upper bound on the convergence rate through the coupling inequality. See Appendix~\ref{sec:coupling} for details.

\begin{remark}
These bounds connect the privacy budget $\epsilon$ to sampler efficiency. Interestingly, stronger privacy (smaller $\epsilon$ and larger noise) leads to faster convergence. This occurs as the data imputation posterior becomes more diffusive and hence easier to approximate by IMH using $q(\cdot) = f_{\theta}(\cdot)$. This is also confirmed by our experiments; see Figure~\ref{fig:syn_acc_compare}.
\end{remark}

\begin{remark}
The analysis we provide in Theorem~\ref{thm:convergence_rate} has focused on $n = 2$.
It is natural to ask whether Theorem~\ref{thm:convergence_rate} can be extended to the case where $n > 2$. Deriving explicit convergence rate upper bounds proves to be difficult for $n \ge 2$. The main challenge lies in controlling the probability of expansion, i.e. $d(x^\prime, \tilde{x}^\prime) = d(x,\tilde{x}) + 1 $. While we provide a bound in Lemma~\ref{lemma:increase}, the result, when used in conjunction with Lemma~\ref{lem:couple_distance}, does not lead to a `contractive coupling' for $n > 2$, in the sense of \citet[Definition 1.9]{blanca2022mixing} and \citet{bubley1997path}.
In fact, a qualitative comparison between Ran- and Sys-Scan Gibbs would also be challenging when $n > 2$, and we refer readers to \citet[Remark 4.4 and Section 6]{qin2022convergence} for more discussions on the technical hurdles.
\end{remark}

\section{NUMERICAL EXPERIMENTS}\label{sec:exp}

In this section, we empirically compare SOMA with Sys-IMwG and Ran-IMwG samplers, 
in terms of convergence rates and acceptance probabilities, which have been studied in Section~\ref{sec:theory}.
We study a differentially private data release scenario in Section~\ref{sec:exp_syn},
and in Section~\ref{sec:exp_linear}, we embed SOMA into a data augmentation MCMC framework for private Bayesian linear regression,
highlighting SOMA's versatility. Moreover, Appendix~\ref {sec:exp_hist} provides an example of data imputation based on perturbed histograms.

\subsection{Imputation Given a Privatized Mean}\label{sec:exp_syn}
We first consider a canonical privacy problem: releasing a privatized mean of confidential data drawn from a bounded domain.
Assuming $X_i \in [0, 1]$ for each $i$, with model them by $X_i \sim \text{Beta}(a_0, b_0)$ independently with $n=2$.
Let $Y \sim \textrm{Laplace}(({X_1+X_2})/{2}, {1}/{\epsilon})$, which satisfy $\epsilon$-DP, be the perturbed mean. 
Given $Y = y$, the posterior density of the confidential data is given by
\begin{align*}\label{eqn:synthetic_beta}
\pi(x\mid y) \propto \ &\exp\left(-\frac{\epsilon}{2}\left|y - \frac{x_1+x_2}{2}\right|\right) \cdot (x_1x_2)^{a_0-{}1}(1-x_1)^{b_0-1}(1-x_2)^{b_0-1}.
\end{align*}
Figure~\ref{fig:syn_trace} gives the 2-dimensional trace plots of the three methods over 500 iterations, starting from the same initial state $(0.3, 0.3)$ with $a_0=b_0=10$ and $\epsilon=20$. 
SOMA traverses the higher density regions of $\pi$ more rapidly and can explore the state space more efficiently.
In contrast, Ran-IMwG gets stuck in some local regions, resulting in an under-exploration of the entire state space.
\begin{figure}[ht]
\centering
\includegraphics[width=0.70\textwidth]{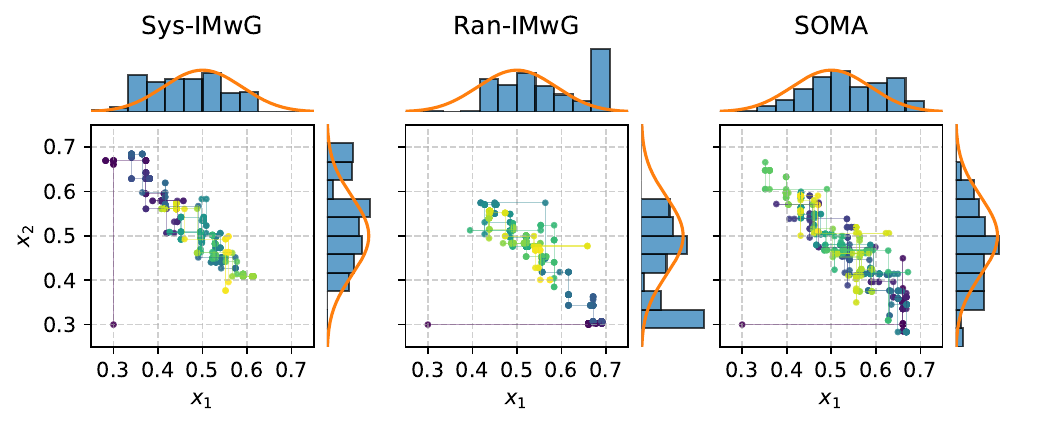}
\caption{Trace plots of different methods. The starting point is marked in deep purple, with subsequent samples in lighter colors. The orange density curve represents the true marginal posterior. SOMA achieves better exploration compared to others.}
\label{fig:syn_trace}
\end{figure}

Next, we quantitatively compare the methods using their empirical acceptance probabilities in Figure~\ref{fig:syn_acc_compare}. 
SOMA has a higher acceptance rate than Ran-IMwG and SYS-IMwG, agreeing with the results proved in Section~\ref{sec:theory}.
\begin{figure}[htbp]
\centering
\begin{subfigure}[t]{0.50\textwidth}
    \vspace{0pt}
    \includegraphics[width = 1\textwidth]{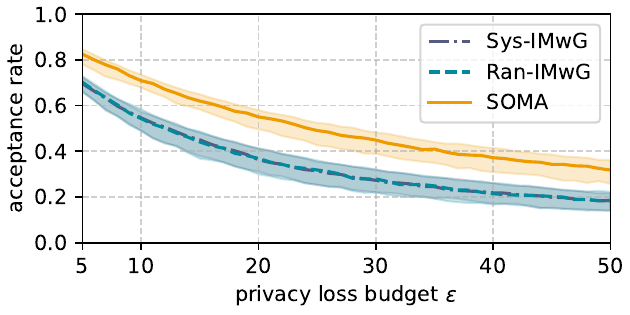}
\end{subfigure}
\caption{
Average acceptance rate comparison for varying values of $\epsilon$. 
}
\label{fig:syn_acc_compare}
\end{figure}

Finally, in challenging scenarios with a strong correlation between $X_1$ and $X_2$, SOMA also exhibits the fastest convergence rate; See Appendix~\ref{appendix:details_syn} for more results.

\subsection{Bayesian Linear Regression Given Privatized Data}\label{sec:exp_linear}

This section concerns the full data augmentation MCMC algorithm targeting the joint distribution $\pi(\theta, x \mid s_\mathrm{dp})$ of \eqref{eqn:joint-posterior} for linear regression. 

The mechanism for releasing perturbed summaries has been described in Example~\ref{example:linearregression}.
We will use conjugate priors for $\theta$, which makes the inference step described in \Cref{sec:imputation} tractable. In this case, the convergence of the marginal $\theta$ chain is dictated mainly by the convergence rate of the imputation step.
Appendix~\ref{appendix:details_pbl} describes the data generating parameters for simulating the dataset.

First, consider the small sample size ($n = 10$) case. The dimension of the joint chain is $n \cdot (p+2) + (p+1) = 43.$
For a small privacy loss budget like $\epsilon = 3$, 
SOMA's acceptance rate is 99.49\%, making it nearly rejection-free, whereas IMwG methods accept around 94.80\%.
SOMA's advantages become more evident under a larger privacy budget ($\epsilon=30$), where it achieves an acceptance rate about 91.91\% (compared to 50\% for IMwG baselines) and couples significantly faster. Full results are available in Appendix~\ref{tab:pbl_coupling1}. 

Next we compare the convergence of the marginal $\theta$ chain at $n = 100.$ 
We no longer use coupling-based diagnosis as the dimension of the state space is $n\cdot(p+2) + (p+1) =403$ at this point.
We use traditional convergence diagnosis such as the $\hat{R}$ (R-hat or the Gelman-Rubin diagnosis) and effective sample size~\citep[ESS]{aki2020rank}. 
A large $\hat{R}$ indicates poorer mixing, while $\hat{R} \approx 1$ is generally considered indicative of good mixing. 
As shown in Figure~\ref{fig:pbl_diagnostic}-A, the chains utilizing SOMA require approximately 800 iterations to achieve $\hat{R}<1.05$, while both Sys-IMwG and Ran-IMwG require about 1,500 iterations. 
As shown in Figure~\ref{fig:pbl_diagnostic}-B, using SOMA within DAMCMC results in a higher ESS than the other two methods.

\begin{figure}[htbp]
\hspace{0.25\textwidth}
\begin{subfigure}[t]{0.50\textwidth}
    \vspace{0pt}
    \includegraphics[width = 1\textwidth]{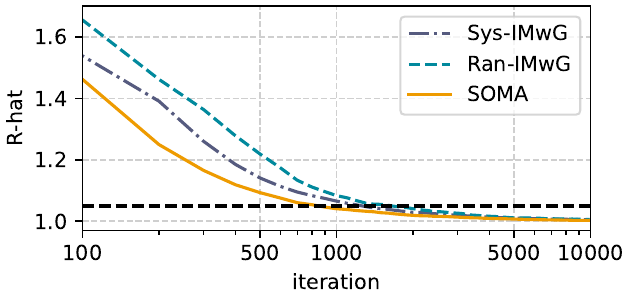}
\end{subfigure}
\hspace{-0.54\textwidth}
\begin{subfigure}[t]{0.02\textwidth}
    \vspace{2pt}
    \textbf{A}
\end{subfigure}

\hspace{0.25\textwidth}
\begin{subfigure}[t]{0.50\textwidth}
    \vspace{2pt}
    \includegraphics[width = 1\textwidth]{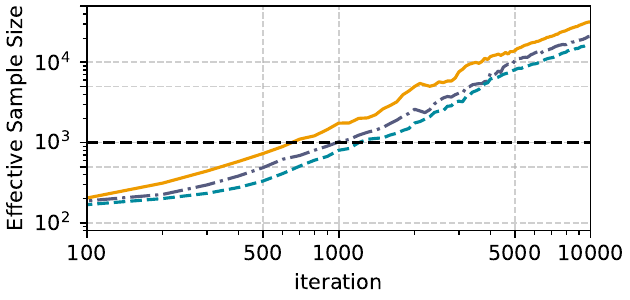}
\end{subfigure}
\hspace{-0.54\textwidth}
\begin{subfigure}[t]{0.02\textwidth}
    \vspace{2pt}
    \textbf{B}  
\end{subfigure}
\caption{
Convergence diagnosis using independent chains.
\textbf{A.} $\hat{R}$ diagnostic across iterations for three methods. The dashed line is 1.05.
\textbf{B.} ESS across iterations for three methods
}
\label{fig:pbl_diagnostic}
\end{figure}

\section{CENSUS DATA EXPERIMENT}\label{sec:realdata}
We apply SOMA to the privacy-preserving analysis of compositional data using the 2019 ATUS microdata file, which is collected by the U.S. Bureau of Labor Statistics. The data is available at \url{https://www.bls.gov/tus/datafiles-2019.htm}. 
The confidential data $\{x_i\}_{i=1}^n$ represent
the fraction of time spent on personal care ($x_{i1}$), eating and drinking ($x_{i2}$), and all other activities ($x_{i3}$). The privacy mechanism has been reviewed in Example~\ref{example:compositionaldata}.
We use a privacy budget $\epsilon=10$ and a clipping lower bound $a=0.0006$.

The data is split by female and male. 
After data pre-processing, the sample consists of 3,528 females and 3,128 males. 
We model the data with a Dirichlet distribution with parameter vector $\alpha = (\alpha_1, \alpha_2, \alpha_3)$ and specify a $\mathrm{Gamma}(1.0, 0.1)$ prior for each $\alpha_j$, $j=1,2,3$.
The same model, budget, clipping lower bound, and priors have also been used in \citet{guo2024differentially,awan2024statistical}, 

Unlike in Section~\ref{sec:exp_linear}, in this problem, the inference step $p(\alpha \mid x_1, \ldots, x_n)$ cannot be sampled exactly, and we resort to using a slice sampler \citep{neal2003slice} to sample from it approximately. 
We compare using SOMA or IMwG for the imputation step.

Table~\ref{tab:atus_diagnostics} presents the chain diagnostics, comparing using SOMA versus using IMwG for the imputation step.
We report the $\hat{R}$ after $10^5$ iterations and the ESS after $10^6$ iterations based on 20 independent runs, which shows that the SOMA-based update achieves a lower $\hat{R}$ and higher ESS across all dimensions.

\begin{table}[tb]
\centering
\caption{Diagnostic Summary for SOMA and IMwG}
\label{tab:atus_diagnostics}
\begin{tabular}{@{}cc cccc @{}}
\multirow{2}{*}{} & \multirow{2}{*}{Method} & \multicolumn{2}{c}{Female} & \multicolumn{2}{c}{Male} \\
\cmidrule(lr){3-4} \cmidrule(lr){5-6}
& & $\hat{R}$ & ESS & $\hat{R}$ & ESS \\
\midrule
\multirow{2}{*}{$\alpha_1$} & SOMA & \textbf{1.0469} & \textbf{2754.5} & \textbf{1.0492} & \textbf{3406.8} \\
                            & IMwG & 1.1041 & 2373.5 & 1.1994 & 2249.0 \\
\midrule
\multirow{2}{*}{$\alpha_2$} & SOMA & \textbf{1.0333} & \textbf{3576.5} & \textbf{1.0344} & \textbf{4510.3} \\
                            & IMwG & 1.0816 & 2965.0 & 1.1546 & 2777.4 \\
\midrule
\multirow{2}{*}{$\alpha_3$} & SOMA & \textbf{1.0481} & \textbf{2713.6} & \textbf{1.0506} & \textbf{3342.5} \\
                            & IMwG & 1.1068 & 2348.7 & 1.2043 & 2216.6 \\
\end{tabular}
\end{table}

Table~\ref{tab:mean_estimates} presents the estimated mean time allocations, $E[X \mid s_{\mathrm{dp}}]$. 
We use these posterior estimates to test for significant gender-based differences, evaluating the null hypothesis that the absolute difference in expected time allocation for each activity $j$ is no more than one percentage point ($H_{0,j}: |E[X_j|\text{Female}] - E[X_j|\text{Male}]| \le 0.01$). 
Our analysis leads to the rejection of the null hypotheses for personal care ($H_{0,1}$) and other activities ($H_{0,3}$), while the null hypothesis for eating and drinking ($H_{0,2}$) is not rejected. 
Both the point estimates and hypothesis testing conclusions are consistent with the prior findings in \citet{guo2024differentially}.

\begin{table}[tb]
\centering
\caption{Estimated Mean Time Allocation $E[X \mid s_{\mathrm{dp}}]$}
\vspace{-5pt}
\label{tab:mean_estimates}
\begin{tabular}{@{}lcccc@{}}
\multirow{2}{*}{Activity} & \multicolumn{2}{c}{Female} & \multicolumn{2}{c}{Male} \\
\cmidrule(lr){2-3} \cmidrule(lr){4-5}
 & SOMA & IMwG & SOMA & IMwG \\
\midrule
Personal Care        & 0.4112 & 0.4112 & 0.3919 & 0.3919 \\
Eating \& drinking  & 0.0507 & 0.0506 & 0.0508 & 0.0508 \\
Other activities     & 0.5381 & 0.5382 & 0.5573 & 0.5574 \\
\end{tabular}
\end{table}

\bibliographystyle{apalike}
\bibliography{refs}

\appendix

\section{PROOFS AND TECHNICAL RESULTS}
\subsection{Proofs of Section \ref{sec:soma}}\label{sec:proofs_of_section_3}
\begin{proof}[Proof of Theorem \ref{thm:rev}]
We can see that the proposal consists of two steps: proposing a new entry $y$ and choosing an index $i$ among $\{1,\ldots,n\}$. This proposal moves $x$ by at most one entry, so $K(x^\prime\mid x)=0$ whenever $d(x,x^\prime) > 1$. 
Without loss of generality, let's assume $x$ and $x'$ differ only on component 1, i.e., $x^\prime=[x_{-1},y]$.

To verify the detailed balance condition, we should derive the `backward proposal kernel' $K(x \mid x^\prime)$. In this proposal kernel, the current state is $x'$, the proposed point is $x_1$, the index $i$ is selected according to weights $w_i(x_1, x')$ defined in \eqref{eqn:wi} \& \eqref{eqn:w0}. Notice that from $x'$, we can only move to $[x'_{-1},x_1] = x$  or $[x^{\prime}_{-i}, x_1]$ which is equivalent to $[x_{-i}, y]$ for $i \ge 2$ up to shuffling indices, in terms of having equal density under $\pi$. 
Using our simplified weight definitions, we have
\begin{align*}
w_1(x_1, x')  &=\eta(s_\mathrm{dp} \mid s([x'_{-1}, x_1])) = \eta(s_\mathrm{dp} \mid s(x)) = w_0(y,x),\\
w_k(x_1,x') &=\eta(s_\mathrm{dp} \mid s([x'_{-k},x_1])) = \eta(s_\mathrm{dp} \mid s([x_{-k}, y])) = w_k(y, x),  \quad \text{for } k \ge 2.
\end{align*}
Using these weights, the backward proposal kernel density is 
\begin{equation*}
K(x \mid x^\prime)=f(x_1\mid\theta)\frac{w_1(x_1, x')}{\sum_{k=1}^n w_k(x_1, x')} = f(x_1\mid\theta)\frac{w_0(y,x)}{w_0(y,x) + \sum_{k=2}^n w_k(y,x)}.
\end{equation*}
We also have 
\begin{equation}\label{eqn:detailed-balance}
\frac{\pi(x^\prime) K(x \mid x^\prime)}{\pi(x) K(x^\prime \mid x)} = 
\frac{\pi(x^\prime) f(x_1\mid\theta)\frac{w_0(y,x)}{w_0 + \sum_{j=2}^n w_j(y,x)}}{\pi(x) f(y\mid\theta)\frac{w_1}{\sum_{j=1}^n w_j(y,x)}}= \frac{\sum_{j=1}^n w_j(y,x)}{w_0 + \sum_{j=2}^n w_j(y,x)},
\end{equation}
now the following identity holds:
\begin{align*}
\pi(x')f(x_1\mid\theta)w_0(y,x) &= \left( \eta(s_\mathrm{dp}\mid s(x')) \prod_{j\neq 1} f(x_j\mid\theta) f(y\mid\theta) \right) f(x_1\mid\theta) \eta(s_\mathrm{dp}\mid s(x)) \\
&= \left( \eta(s_\mathrm{dp}\mid s(x)) \prod_{j\neq 1} f(x_j\mid\theta) f(x_1\mid\theta) \right) f(y\mid\theta) \eta(s_\mathrm{dp}\mid s(x')) \\
&= \pi(x)f(y\mid\theta)w_1(y,x).
\end{align*}
Plugging \eqref{eqn:detailed-balance} into the Metropolis-Hastings acceptance ratio
$
\alpha (x^\prime \mid x) = \min \left(1, \frac{\pi(x^\prime) K(x \mid x^\prime)}{\pi(x) K(x^\prime \mid x)}\right),$
yields \eqref{eqn:acceptprob}, which makes the Markov transition kernel reversible with repsect to $\pi.$
\end{proof}

\begin{proof}[Proof of Corollary \ref{coro:sufficientconditions}]
For the general cases, assume $x$ and $x'$ differ only on component 1, i.e., $x^\prime=[x_{-1},y]$.
The weights are now given by
\begin{align*}
w_1(x_1, x')  &=c\cdot\frac{\pi([x'_{-1}, x_1])}{q(x_1)\prod_{j=2}^n q(x_j^{\prime})}=c\cdot\frac{\pi(x)}{\prod_{j=1}^n q(x_j)} = w_0(y,x),\\
w_k(x_1,x') &=c\cdot\frac{\pi([x'_{-k},x_1])}{q(x_1)q(y)\prod_{2 \le j \le n, j \not=k}q(x^{\prime}_j)} = c\cdot\frac{\pi([x_{-k}, y])}{q(y)\prod_{1 \le j \le n, j \not=k}q(x_j)} = w_k(y, x),  \quad \text{for } k \ge 2,
\end{align*}
and we have the corresponding backward proposal kernel as
\begin{equation*}
K(x \mid x^\prime)=q(x_1)\frac{w_1(x_1, x')}{\sum_{k=1}^n w_k(x_1, x')} = q(x_1)\frac{w_0(y,x)}{w_0(y,x) + \sum_{k=2}^n w_k(y,x)},
\end{equation*}
which yields
\begin{equation*}
\frac{\pi(x^\prime) K(x \mid x^\prime)}{\pi(x) K(x^\prime \mid x)} = 
\frac{\pi(x^\prime) q(x_1)\frac{w_0(y,x)}{w_0 + \sum_{j=2}^n w_j(y,x)}}{\pi(x) q(y)\frac{w_1}{\sum_{j=1}^n w_j(y,x)}}= \frac{\sum_{j=1}^n w_j(y,x)}{w_0 + \sum_{j=2}^n w_j(y,x)}.
\end{equation*}
Now since
\begin{equation*}\pi(x')q(x_1)w_0(y,x) = \pi(x')q(x_1)\frac{c\pi(x)}{\prod_{j=1}^n q(x_j)}  = c\pi(x) \frac{q(y)}{q(y)} \frac{\pi([x_{-1},y])}{\prod_{j=2}^n q(x_j)}  = \pi(x) q(y) w_1(y,x),
\end{equation*}
thus, the Metropolis-Hastings acceptance ratio
$
\alpha (x^\prime \mid x) = \min \left(1, \frac{\pi(x^\prime) K(x \mid x^\prime)}{\pi(x) K(x^\prime \mid x)}\right)$
still can yields \eqref{eqn:acceptprob}, which makes the Markov transition kernel reversible with respect to $\pi$.
\end{proof}

\subsection{Generalization to SOMA with $m < n$ components}
\label{app:soma_on_subset}

We now provide a detailed discussion on the properties of the SOMA sampler when adapted to operate on a subset $M \subseteq \{1, \dots, n\}$ of size $m$.

First, when $m = 1$, SOMA becomes deterministic or random scan, depending on how the index set is chosen.

Second, each SOMA kernel using $m$ components still satisfies detailed balance. 
\begin{lemma}\label{lem:soma_detailed_balanced_on_subset}
Use same conditions as \Cref{coro:sufficientconditions}.
Let $P(x' \mid x, M)$ be the transition probability for a given index set $M$. 
The SOMA sampler operating on $M$ satisfies the detailed balance condition with respect to $\pi(x)$:
\begin{equation}\label{eq:detailed_balance_on_M}
\pi(x) P(x' \mid x, M) = \pi(x') P(x \mid x', M).
\end{equation}
\end{lemma}
\begin{proof}
We verify detailed balance for $x \not=x'$ and this is to check 
$\pi(x) K(x' \mid x, M) \alpha(x' \mid x, M) = \pi(x') K(x \mid x', M) \alpha(x \mid x'. M) $ Let $i$ be the index where $x_i \not= x^{\prime}_i$. If $i \not\in M$, then both sides of \eqref{eq:detailed_balance_on_M} are zero. 
The proposal density for moving from $x$ to $x'$ (by proposing $y = x'_i$ and selecting index $i \in M$) has the form
\begin{equation*}
K(x' \mid x, M) = q(x_i') \frac{w_i(x'_i,x)}{W_{M}(y,x)}, \quad \text{where }\ W_{M}(x'_i,x) = \sum_{j \in M} w_j(x'_i,x),
\end{equation*}
and the backward move from $x'$ to $x$ involves proposing the $x_i$ to replace $x_i'$, which yields
\begin{equation*}
K(x \mid x', M) = q(x_i) \frac{w_i(x_i, x')}{W'_{M}(x_i, x')}, \quad \text{where }\ W'_{M}(x_i, x') = \sum_{j \in M} w_j(x_i, x').
\end{equation*}
Recall also that 
\begin{align*}
\alpha(x' \mid x, M) &=  \min\left(1, \frac{W_{M}(x'_i,x)}{w_0(x'_i,x) + \sum_{j\not=i, j \in M} w_j(x'_i, x)}\right) = \min\left(1, \frac{W_{M}(x'_i,x)}{w_0(x'_i,x)  + W_{M}(x'_i,x) - w_i(x'_i,x)}\right), 
\end{align*}
By construction of the weights \eqref{eqn:wi} and \eqref{eqn:w0}, we have $w_0(x'_i, x) = w_i(x_i,x')$ and $w_0(x_i,x') = w_i(x'_i, x).$ More importantly $w_j(x_i', x) = w_j(x_i, x')$ for each $j \in M$ due to permutation invariance.
Now we have 
\begin{align*}
W'_{M}(x_i, x') &= w_i(x_i, x') + \sum_{j \not=i,j \in M} w_j(x_i, x') \\&= w_0(x'_i, x) + \sum_{j\not=i, j \in M}w_j(x'_i,x) \\&= w_0(x'_i, x)  - w_i(x'_i, x) + W_M(x'_i,x).
\end{align*}
We can check that
\begin{align*}
\frac{\pi(x') K(x \mid x', M)}{\pi(x) K(x' \mid x, M)}&=  \frac{\pi(x') \cdot q(x_i) \frac{w_i(x_i, x')}{W'_{M}(x_i, x')}}{\pi(x)\cdot  q(x'_i) \frac{w_i(x'_i,x)}{W_{M}(x'_i,x)}} = \frac{\pi(x') \cdot q(x_i) w_i(x_i, x') W_{M}(x'_i,x)}{\pi(x)\cdot  q(x'_i) w_i(x'_i,x) W'_{M}(x_i, x')}  \\
&= \frac{\pi(x') q(x_i) \cdot w_0(x'_i,x) }{\pi(x) q(x'_i) \cdot w_0(x_i,x') } \cdot \frac{W_M(x'_i, x)}{w_0(x'_i) - w_i(x'_i,x) + W_M(x'_i,x)} \\
&= \frac{W_{M}(x'_i,x)}{W_{M}(x'_i,x) + w_0(x'_i,x) - w_i(x'_i,x)}.
\end{align*}
Setting $\alpha(x' \mid x, M) = \min\left(1, \frac{\pi(x') K(x \mid x', M)}{\pi(x) K(x' \mid x, M)}\right)$ satisfies detailed balance and hence makes the Markov transition kernel reversible with respect to $\pi.$
\end{proof}

Finally, 
SOMA remains reversible with respect to the target distribution $\pi(x)$ when using randomly selected subsets. Assume that at each iteration, a subset $M$ of size $m$ is drawn uniformly at random from all possible subsets of $\{1, \dots, n\}$ of that size.
The overall transition kernel $P_m(x' \mid x)$ becomes
\begin{equation*}
P_m(x' \mid x) = \frac{1}{\binom{n}{m}} \sum_{M} P(x' \mid x, M).
\end{equation*}

With Lemma~\ref{lem:soma_detailed_balanced_on_subset}, the reversibility for the random subset case is straightforward by verifying the detailed balance equation as
\begin{equation*}
\pi(x) P_m(x' \mid x) = \frac{1}{\binom{n}{m}}\pi(x) \sum_{M} P(x' \mid x, M) = \frac{1}{\binom{n}{m}}\pi(x') \sum_{M}P(x \mid x', M) = \pi(x') P_m(x \mid x').
\end{equation*}
Thus, the SOMA sampler with randomly chosen subsets is still reversible with respect to the target distribution $\pi(x)$.

\subsection{SOMA as a multi-proposal sampler}\label{sec:multiproposal-1}
In Section~\ref{sec:soma}, we mention that SOMA can be viewed as a special case of multi-proposal MCMC methods. We provide a more detailed discussion here.

We present another proof of Theorem~\ref{thm:peskun}.
\begin{proof}
SOMA simultaneously considers $n$ proposed states $[x_{-1}, y]$, $[x_{-2}, y], \ldots, [x_{-n}, y]$, with (unnormalized) weights $w_i(y, x).$ 
In this view, SOMA is indeed a multi-proposal MCMC method. 
The Markov transition kernel of SOMA  $P_{\mathrm{SOMA}}(\cdot \mid x)$ is a special case of the $P^{(K)}( \cdot \mid x)$ from \citet{pozza2025fundamental} with $K = n$. 
SOMA has a special way of constructing multiple proposals and a weight function to choose among the proposed states. 

In SOMA, the `marginal proposal' for each component $i$ is in fact the IMH proposal $f(\cdot\mid\theta)$. Let 
\begin{equation*}
Q_i(x' \mid x) = \begin{cases}
    q(x'_i) & \text{if } x'_j = x_j \quad \forall j \not= i \\
    0 & \text{otherwise}
\end{cases}
\end{equation*}
denote the IMH proposal kernel for the $i$-th component, defined on the full state space  $\mathbb{X}^n.$
Then $\widetilde{Q}(\cdot \mid x) = \frac{1}{n} Q_i(\cdot \mid x)$
is the overall proposal kernel for Ran-IMwG: it proposes an index $i$ uniformly and proposes a new from $\mathbb{X}$.
We denote the Markov transition kernel associated with proposal $\widetilde{Q}$ and 
acceptance probability \eqref{eqn:imh-acceptance} by $\widetilde{P}.$
It coincides with the Markov transition kernel of Ran-IMwG, denoted $P_{\mathrm{RAN}}(\cdot \mid x).$
In other words, Ran-IMwG is the `single proposal' kernel that we should compare SOMA with. \citet[Theorem 1]{pozza2025fundamental} states that 
\begin{equation*}
P^{(n)}(A\setminus \{x\} \mid  x) \le n \widetilde{P}(A \setminus \{x\} \mid x)
\end{equation*}
for $\pi$-almost all $x$ and for all measurable $A \subset \mathbb{X}$. 
In the context of SOMA and Ran-IMwG, this is 
\begin{equation*}
P_{\mathrm{SOMA}}(A \setminus \{x\} \mid x) \le n P_{\mathrm{RAN}}(A \setminus \{x\} \mid x),
\end{equation*}
which is equivalent to \eqref{eqn:peskun-x} in Theorem~\ref{thm:peskun}.
\end{proof}

\subsection{Asymptotic Variance of SOMA and Computation Cost}
\label{sec:cost}
A consequence of Theorem~\ref{thm:peskun} is that, a \emph{serial} implementation of SOMA would not be competitive against Ran-IMwG in terms of statistical efficiency once we consider their computation costs.

We can measure efficiency with the asymptotic variance in a central limit theorem for an MCMC estimator~\citep{jones2004markovchainclt}: given a (measurable) test function $h$ and a kernel $P$ invariant to $\pi$, our estimand is $\mathbb{E}_{\pi}(h(X))$ and its variance is $\mathrm{Var}_{\pi}(h(X)) < \infty$. Using MCMC samples drawn from $P$, we denote the asymptotic variance by $\mathrm{Var}(h, P)$ \citep{neal2004improving}.
\begin{corollary}\label{corollary:limitation}
For any function $h$ integrable with respect to $\pi$, we have
\begin{equation}\label{eqn:asymptoticvariance}
     \mathrm{Var}(h, P_{\mathrm{RAN}}) \le n \mathrm{Var}(h, P_{\mathrm{SOMA}}) + (n-1) \mathrm{Var}_{\pi}(h).
\end{equation}
\end{corollary}
\begin{proof}
    This result follows from Theorem~\ref{thm:peskun} and \citet[Lemma 33]{andrieu2018uniform}.
\end{proof}

Since RanScan involves only ${w_i}/{w_0}$ for a specific $i$ but SOMA requires $\{w_j\}_{0 \le j \le n}$, the cost of SOMA could be $\mathcal{O}(n)$ that of RanScan with a serial implementation. When $nT$ iterations of RanScan has the same runtime with $T$ iterations of SOMA, and these estimators have variances $(nT)^{-1}\mathrm{Var}(h, P_{\mathrm{RAN}})$ and $T^{-1}\mathrm{Var}(h, P_{\mathrm{SOMA}})$ respectively. Equation \eqref{eqn:asymptoticvariance} suggests $nT$ iterations of RanScan are asymptotically more efficient than $T$ iterations of SOMA, implemented serially.

Luckily, computing the weights in parallel can reduce the runtime cost of SOMA. 
Oftentimes, the target distribution $\pi$ has the structure to make this parallel/vectorized implementation feasible.
SOMA can thus have a computational complexity comparable to Ran-IMwG and Sys-IMwG under moderate values of $n$, using a vectorized implementation. The costs would have the same order under a parallel implementation.
These allow one to harness the efficiency gained from SOMA's design.
We give Example~\ref{ex:additive_summary}, which continues with the perturbed summary statistics privacy mechanism in Section~\ref{sec:exp_linear}.

\begin{example}[Record additive privacy mechanisms]\label{ex:additive_summary}
Privacy mechanisms that add data-independent noise to a query of the form $s(x) = \sum_{i=1}^n t(x_i)$ have densities
$\eta(s_{\mathrm{dp}} \mid s(x)) = g(s_{\mathrm{dp}}, \sum_{i=1}^n t_i(x_i))$ for tractable and known functions $g$ and $t$. 
\citet[Assumption 2]{ju2022data} coined the term `record additivity' to describe how $\eta(s_{\mathrm{dp}} \mid s(x))$ depends on an additive summary of $x$, and connects this property with implementation efficiency of data augmentation methods.
With the widely used Laplace mechanism $s_{\mathrm{dp}} \sim \mathrm{Laplace}(s(x), {\Delta}/{\epsilon})$ and the weights $w_i = \eta(s_{\mathrm{dp}} \mid s([x_{-i}, y]))$, we have 
\begin{equation}\label{eqn:laplace}
    \log\left(\frac{w_i}{w_0}\right) = \frac{\Delta}{\epsilon} \left(\left\|s_{\mathrm{dp}} - \sum_{j=1}^n t(x_j)\right\|_1 -  \left\|s_{\mathrm{dp}} - \sum_{j=1,j \neq i}^n t(x_j) - t(y)\right\|_1\right).
\end{equation}
In \eqref{eqn:laplace}, $\Delta$ is a sensitivity parameter related to $t(\cdot)$ and $\epsilon$ is the privacy loss budget.
\end{example}

\Cref{example:perturbedhistogram,example:compositionaldata,example:linearregression} are all instances of record addivite privacy mechanisms.

The swap-one-out structure in \eqref{eqn:laplace} enables parallel or vectorized implementation, which can reduce SOMA’s computational complexity to a large extent.
We include empirical run time comparisons in Section~\ref{sec:runtime}, showing the runtime cost of SOMA is less than twice that of the Ran/Sys-IMwG methods with $n = 100$ components. In this case, instead of comparing the asymptotic variance of using $T$ iterations of SOMA with $nT$ iterations of Ran-IMwG, we should use $2T$ iterations of Ran-IMwG.

\subsection{Proofs of Section~\ref{sec:acc_bound}}\label{sec:proofs_of_section_41}
\begin{proof}[Proof of Theorem~\ref{thm:acc_compare}]
Consider the Metropolis acceptance probability for updating the $i$-th component of $x \in \mathbb{X}^n$ with $y \in \mathbb{X}$ when index $i$ is selected. For the SOMA sampler, this probability is given by  
\begin{equation*}
\alpha_i^{\mathrm{SOMA}}(y,x)=\min\left\{1, \frac{W(y, x)}{W(y, x) + w_0(y, x) - w_i(y, x)}\right\}=\min\left\{1, \frac{W(y, x) - w_i(y, x) + w_i(y, x)}{W(y, x) - w_i(y, x) + w_0(y, x)}\right\}.
\end{equation*}
Since $W(y, x) - w_i(y, x) = \sum_{j\neq i}^n w_j(y, x)\ge 0$, we have
\begin{equation*}
\alpha_i^{\mathrm{SOMA}}(y,x)\ge \min\left\{1, \frac{w_i(y,x)}{w_0(y,x)}\right\} = \alpha^{\mathrm{IMwG}}_i (y,x),
\end{equation*}
which ends the proof.
\end{proof}

\begin{proof}[Proof of Theorem~\ref{thm:peskun}]
Suppose $x$ and $x'$ differ on component $i$ and, without loss of generality, assume $i = 1$. Then, starting from the right-hand side, 
\begin{align*}
    &\quad n \cdot P_{\mathrm{RAN}}(x' \mid x) \\
    &= n \cdot \frac{1}{n} f(x'_1\mid\theta)\alpha_1^{\mathrm{IMwG}}(x'_1,x)  =n \cdot \frac{1}{n} f(x'_1\mid\theta)\min\left(\frac{w_1}{w_0}, 1\right) \\
    &= f(x'_1\mid\theta) \cdot w_1 \cdot \min\left(\frac{1}{w_1}, \frac{1}{w_0}\right) \\
    & \ge f(x'_1\mid\theta)\cdot  w_1  \cdot \min\left(\frac{1}{w_1 + \sum_{i=2}^n w_i}, \frac{1}{w_0 + \sum_{i=2}^n w_i}\right) \\
    &= f(x'_1\mid\theta)\cdot \frac{w_1}{w_1 + \sum_{i=2}^n w_i} \min\left( 1, \frac{w_1 + \sum_{i=2}^n w_i}{w_0 + \sum_{i=2}^n w_i}\right) \\
    &= P_{\mathrm{SOMA}}(x' \mid x).
\end{align*}
\end{proof}

\begin{lemma}\label{lemma:acc_compare}
For any $x\in\mathbb{X}^n$ and $y\in\mathbb{X}$, the following inequality holds:
\begin{equation}
\sum_{i=1}^{n}  \left(\frac{w_i(y,x)}{W(y,x)} \right) \cdot \alpha_i^{\mathrm{SOMA}}(y,x) \ge \sum_{i=1}^{n} \left( \frac{1}{n}\right) \cdot \alpha_i^{\mathrm{IMwG}}(y,x).
\end{equation}
\end{lemma}

\begin{proof}[Proof of Lemma~\ref{lemma:acc_compare}]
We begin with the base case where $n=2$. Pick any $x\in\mathbb{X}^2$ and $y\in\mathbb{X}$. We aim to show
\begin{equation}
\sum_{i=1}^{2}  \left(\frac{w_i(y,x)}{W(y,x)} \right) \cdot \alpha_i^{\mathrm{SOMA}}(y,x) \ge \sum_{i=1}^{2} \left( \frac{1}{2}\right) \cdot \alpha_i^{\mathrm{IMwG}}(y,x).
\end{equation}
For brevity, define $w_i=w_i(y, x)$ for $i=0,1,2$. Rewriting both sides:
\begin{align*}
\mathrm{LHS}&=\sum_{i=1}^{2}  \frac{w_i}{w_1+w_2} \cdot \min\left\{ 1,\frac{w_1+w_2}{w_0+w_1+w_2-w_i} \right\} \\
& = \min\left\{ \frac{w_1}{w_1+w_2},\frac{w_1}{w_0+w_2} \right\} + \min\left\{ \frac{w_2}{w_1+w_2},\frac{w_2}{w_0+w_1} \right\},
\end{align*}
and
\begin{align*}
\mathrm{RHS}&=\frac12 \left(\min\left\{ 1,\frac{w_1}{w_0} \right\} + \min\left\{ 1,\frac{w_2}{w_0} \right\}\right).
\end{align*}

Without loss of generality, we assume $w_1 \le w_2$. We split into three cases based on the relative sizes of $(w_0,w_1,w_2)$.

\noindent \textbf{Case (a)}. $w_0 \le w_1 \le w_2$.
In this scenario, $\mathrm{LHS} = \mathrm{RHS}=1$, so the inequality holds trivially.

\noindent \textbf{Case (b)}. $w_1 \le w_0 \le w_2$.
In this case, we have 
\begin{align*}
\mathrm{LHS} - \mathrm{RHS}&=
\left(\frac{w_1}{w_0+w_2} + \frac{w_2}{w_1+w_2}\right) - \frac{1}{2}\left( \frac{w_1}{w_0} + 1 \right) \\
&=\frac{w_2\left(w_0^2-w_1^2\right)+\left(w_0-w_1\right)\left(w_2^2-w_0w_1\right)}{2w_0\left(w_0+w_2\right)\left(w_1+w_2\right)}\\
&\ge 0,
\end{align*}
hence $\mathrm{LHS} \ge \mathrm{RHS}$.

\noindent \textbf{Case (c)}. $w_1 \le w_2 \le w_0$.
By a similar calculation,
\begin{align*}
\mathrm{LHS} - \mathrm{RHS}&=
\left(\frac{w_1}{w_0+w_2} + \frac{w_2}{w_0+w_1}\right) - \frac{1}{2}\left( \frac{w_1}{w_0} + \frac{w_2}{w_0} \right) \\
&=\frac{w_0\left(w_1-w_2\right)^2+\left(w_1+w_2\right)\left(w_0^2-w_1w_2\right)}{2w_0\left(w_0+w_1\right)\left(w_0+w_2\right)}\\
&\ge 0,
\end{align*}
thus $\mathrm{LHS} \ge \mathrm{RHS}$ here as well.

In all cases, $\mathrm{LHS} \ge \mathrm{RHS}$ holds for $n=2$.

Now, we extend the argument to general $n$. 
Let $w_0,w_1,\ldots,w_n$ be positive weights. 
Without loss of generality, assume $w_0=1$ and $w_1\le w_2\le\ldots\le w_{k} \le w_0 \le w_{k+1} \le\ldots \le w_n$, and denote $W=w_1+\ldots+w_n$. 
Then, we express the two sides as
\begin{align*}
\mathrm{LHS}&=\sum_{i=1}^{k}\left( \frac{w_i}{W+w_0-w_i} \right)+\sum_{i=k+1}^n \frac{w_i}{W},\\
\mathrm{RHS}&=\sum_{i=1}^k \frac{w_i}{nw_0}+\sum_{i=k+1}^n \frac{1}{n}.
\end{align*}
Note that $\sum_{i=1}^n w_i/W = 1$ and $\sum_{i=1}^n 1/n = 1$, and define $z_i=1-w_i$ for all $i=1,\ldots,k$, we rewrite
\begin{align}
1-\mathrm{LHS}&=\sum_{i=1}^{k}\left( \frac{w_i}{W}-\frac{w_i}{W+1-w_i} \right)
=\sum_{i=1}^{k}\left( \frac{w_i(1-w_i)}{W(W+1-w_i)} \right)
=\sum_{i=1}^{k}\left( \frac{z_i(1-z_i)}{W(W+z_i)} \right),\label{eq:one_minus_lhs} \\
1-\mathrm{RHS}&=\sum_{i=1}^k \left( \frac{1}{n} - \frac{w_i}{n} \right) \
=\sum_{i=1}^{k}\left( \frac{1-w_i}{n} \right) 
=\sum_{i=1}^{k}\left( \frac{z_i}{n} \right).\label{eq:one_minus_rhs}
\end{align}
Next, we consider the following cases: $k\le n-\sqrt{n}$, $k \in (n-\sqrt{n},n-1]$ and $k = n$.

\noindent \textbf{Case (a)}. We assume $k\le n-\sqrt{n}$. Using this bound, we have
\begin{equation*}
W=\sum_{i=1}^k w_i + \sum_{j=k+1}^n w_j
\ge \sum_{j=k+1}^n w_j \ge \sum_{j=k+1}^n w_0 = n-k \ge \sqrt{n},
\end{equation*}
based on Equations \eqref{eq:one_minus_lhs} and \eqref{eq:one_minus_rhs}, since $z_i\in[0, 1]$ for all $i=1,\ldots,n$, we bound
\begin{equation*}
1-\mathrm{LHS} = \sum_{i=1}^{k}\left( \frac{z_i(1-z_i)}{W(W+z_i)} \right)
\le \sum_{i=1}^{k}\left( \frac{z_i(1-z_i)}{W\times W} \right)
\le \sum_{i=1}^{k}\left( \frac{z_i(1-z_i)}{n} \right)
\le \sum_{i=1}^{k}\left( \frac{z_i}{n} \right) = 1-\mathrm{RHS},
\end{equation*}
it gives $\mathrm{LHS} \ge \mathrm{RHS}$.

\noindent \textbf{Case (b)}. We assume $k \in (n-\sqrt{n},n-1]$. Denote $S=\sum_{i=1}^k z_i \le k \le n-1$, we have the following inequality:
\begin{equation}\label{eq:ineq_n-s}
W=\sum_{i=1}^k (1-z_i) + \sum_{j=k+1}^n w_j
\ge \sum_{i=1}^k (1-z_i) + \sum_{j=k+1}^n 1 = n - \sum_{i=1}^k z_i=n-S > 0,
\end{equation}
and
\begin{equation}\label{eq:ineq_ess}
\sum_{i=1}^k z_i^2\ge \frac{\left( \sum_{i=1}^k z_i \right)^2}{k} = \frac{S^2}{k}.
\end{equation}
Based on Equations \eqref{eq:ineq_n-s} and \eqref{eq:ineq_ess}, it holds that
\begin{align*}
1-\mathrm{LHS}&=\sum_{i=1}^{k}\left( \frac{z_i(1-z_i)}{W(W+z_i)} \right)
\le \sum_{i=1}^{k}\left( \frac{z_i(1-z_i)}{(n-S)^2} \right)
=\frac{\sum_{i=1}^{k}z_i-\sum_{i=1}^{k}z_i^2}{(n-S)^2}
\le \frac{S - \frac{S^2}{k}}{(n-S)^2}.
\end{align*}
Since $1-\mathrm{RHS}=S/n$, to prove $\mathrm{LHS} \ge \mathrm{RHS}$, we only need to show
\begin{equation}
\text{LHS} - \text{RHS} \ge \frac{S}{n} - \frac{S - \frac{S^2}{k}}{(n-S)^2} = \frac{S}{kn(n-S)^2}\times \left[nS+k\left( (n-S)^2-n \right)\right]\ge 0.
\end{equation}
Let $G = nS+k\left( (n-S)^2-n \right)$. It suffices to show $G \ge 0.$ Base on the fact that $S=\sum_{i=1}^k z_i\in[0, k]$ and $k \in (n-\sqrt{n},n-1]$, we take $\alpha=S/k\in[0,1]$ and $\beta=k/n\in(\frac{n-\sqrt{n}}{n},\frac{n-1}{n}]$. Now $G$ has the form
\begin{equation}
G=n^2\beta\left( n(\alpha\beta-1)^2 + \alpha - 1 \right),
\end{equation}
and 
\begin{equation*}
\frac{\partial G}{\partial \alpha} = n^2\beta\left( 1 + 2n\beta (\alpha\beta-1) \right)\le
n^2\beta\left( 1 + 2n \cdot \left(\frac{n-\sqrt{n}}{n}\right) \left(1\cdot \frac{n-1}{n}-1\right) \right) = -1 + \frac{2\sqrt{n}}{n}.
\end{equation*}
We have $\frac{\partial G}{\partial \alpha} \le 0$ for all $n\ge 4$. 
As for $n=3$, 
\begin{equation*}
\frac{\partial G(\alpha,\beta; n=3)}{\partial \alpha}=9\beta \left( 1+6\beta (\alpha\beta-1) \right)\le 9\beta \left( 6\beta^2 - 6\beta + 1\right) \le 0, \quad \forall \  \beta \in \left(\frac{3-\sqrt{3}}{3},\frac{2}{3}\right].
\end{equation*}
Thus for all $\alpha\in[0, 1]$, $\beta\in(\frac{n-\sqrt{n}}{n},\frac{n-1}{n}]$ and $n\ge 3$,
\begin{equation}
G(\alpha,\beta;n)\ge G(1,\beta;n)=n^2\beta n(\beta-1)^2\ge 0.
\end{equation}
This proves $\mathrm{LHS} \ge \mathrm{RHS}$ under case (b).

\noindent \textbf{Case (c)}. We assume $k = n$. In this case, since $w_0\ge w_i$ holds for all $i=1,\ldots,n$, thus $W-w_i=\sum_{j\neq i}^n w_j \le (n-1) w_0$. We derive
\begin{align*}
\mathrm{LHS}=\sum_{i=1}^n \frac{w_i}{W+w_0-w_i}\ge \sum_{i=1}^n \frac{w_i}{n w_0}=\mathrm{RHS}.
\end{align*}
Thus, the inequality holds in all cases.
\end{proof}

\begin{proof}[Proof of Theorem~\ref{thm:acc_soma_vs_ran}]
Based on Lemma~\ref{lemma:acc_compare}, integrating both sides over $f(y\mid\theta)$ yields
\begin{equation*}
\int_{\mathbb{X}} f(y\mid\theta) \sum_{i=1}^{n}  \left(\frac{w_i(y,x)}{W(y,x)} \right) \cdot \alpha_i^{\mathrm{SOMA}}(y,x)\mathrm{d}y \ge 
\int_{\mathbb{X}} f(y\mid\theta) \sum_{i=1}^{n} \left( \frac{1}{n}\right) \cdot \alpha_i^{\mathrm{IMwG}}(y,x) \mathrm{d}y,
\end{equation*}
which means $A^{\mathrm{SOMA}}(x)\ge A^{\mathrm{RAN}}(x)$ for all $x\in\mathbb{X}^n$.
\end{proof}

\begin{proof}[Proof of Corollary~\ref{corollary:tv-ran-soma}]
Let $x\in\mathbb{X}^n$ and define $\varphi(\cdot)=\mathbb{I}_{x}(\cdot)$, which maps $\mathbb{X}^n$ to $[0, 1]$. Now for any function $\psi:\mathbb{X}^n\rightarrow[0, 1]$ and any $x^\prime\in\mathbb{X}^n$, it holds that
\begin{equation*}
\int_{\mathbb{X}^n} \psi(\nu) \mathrm{d} P(\nu \mid x^\prime) = \int_{\mathbb{X}^n} \psi(\nu) \alpha(x^\prime, \nu)\mathrm{d} K(\nu \mid x^\prime) + [1-A(x^\prime)]\psi(x^\prime)\ge [1-A(x^\prime)]\psi(x^\prime),
\end{equation*}
since the first term is nonnegative. To bound the $t$-step transition, we iteratively apply this inequality. Using the Chapman-Kolmogorov decomposition, it gives
\begin{equation*}
\int_{\mathbb{X}^n} \varphi(\nu) \mathrm{d} P^t(\nu \mid x) = \int_{\mathbb{X}^n} \left( \int_{\mathbb{X}^n} \varphi(v) \mathrm{d} P(v \mid u) \right) \mathrm{d} P^{t-1}(u \mid x).
\end{equation*}
For each $u$, the inner integral satisfies $\int_{\mathbb{X}^n} \varphi(v) \mathrm{d} P(v \mid u) \ge [1 - A(u)]\varphi(u)$. Substituting this bound,
\begin{equation*}
\int_{\mathbb{X}^n} \varphi(\nu) \mathrm{d} P^t(\nu \mid x) \ge \int_{\mathbb{X}^n} [1 - A(u)]\varphi(u) \mathrm{d} P^{m-1}(u \mid x).
\end{equation*}
Since $\varphi(u) = \mathbb{I}_x(u)$, the integral reduces to $[1 - A(x)] \cdot \int_{\mathbb{X}^n} \varphi(u) \mathrm{d} P^{m-1}(u \mid x)$. By induction, we have
\begin{equation*}
\int_{\mathbb{X}^n} \varphi \mathrm{d} P^t(\cdot \mid x) \ge [1 - A(x)]^t \varphi(x).
\end{equation*}
Thus, it holds that
\begin{align*}
\|P^t (\cdot \mid x) - \pi(\cdot) \|_{\mathrm{TV}} &= \sup_\psi \left|\int_{\mathbb{X}^n} \psi\mathrm{d} P^t (\cdot \mid x) - \int_{\mathbb{X}^n} \psi\mathrm{d} \pi(\cdot)\right| \\
&\ge \left|\int_{\mathbb{X}^n} \varphi\mathrm{d} P^t (\cdot \mid x) - \int_{\mathbb{X}^n} \varphi\mathrm{d} \pi(\cdot)\right| \\
& = \int_{\mathbb{X}^n} \varphi\mathrm{d} P^t (x, \cdot) \\
&\ge [1-A(x)]^t\varphi(x) \\
&=[1-A(x)]^t.
\end{align*}
Hence when $A^{\mathrm{SOMA}}(x) \ge A^{\mathrm{RAN}}(x)$ holds for all $x$, we conclude that
\begin{equation}
\|P^t_{\mathrm{RAN}} (\cdot \mid x) - \pi(\cdot) \|_{\mathrm{TV}} \ge \left(1-A^{\mathrm{RAN}}(x)\right)^t \ge \left(1-A^{\mathrm{SOMA}}(x)\right)^t.
\end{equation}
\end{proof}

\begin{proof}[Proof of Theorem~\ref{thm:acc_bound}]
We first prove that $A^{\mathrm{SOMA}}(x) \ge {n}/({n+e^\epsilon-1})$.
Let $T_i(y, x) = \sum_{j=1,j\neq i}^n w_j(y, x)$. Under the definition of $\epsilon$-DP, we have $w_i(y, x)/w_0(y, x)\ge 1/e^\epsilon$, and $T_i(y, x)/w_0(y, x)\ge (n-1)/e^\epsilon$ holds for all $x$ and $y$. Now the overall acceptance probability $A^{\mathrm{SOMA}}(x)$ for SOMA has the form
\begin{align}
A^{\mathrm{SOMA}}(x) &= \int_{\mathbb{X}} f(y\mid\theta) \sum_{i=1}^n \left[\mathbb{P}(I=i) \min\left\{1, \frac{W(y, x)}{W(y, x)+w_0(y, x)-w_i(y, x)}\right\}\right]\mathrm{d}y \notag \\
&=\int_{\mathbb{X}} f(y\mid\theta) \sum_{i=1}^n\left[ \mathbb{P}(I=i)\min\left\{1, \frac{T_i(y, x) + w_i(y, x)}{T_i(y, x) + w_0(y, x)}\right\}\right] \mathrm{d}y \notag \\
&=\int_{\mathbb{X}} f(y\mid\theta) \sum_{i=1}^n\left[ \mathbb{P}(I=i)\min\left\{1, 1 - \frac{1 - w_i(y, x)/w_0(y, x)}{1 + T_i(y, x)/w_0(y, x)}\right\} \right] \mathrm{d}y \notag \\
&\ge \int_{\mathbb{X}} f(y\mid\theta) \sum_{i=1}^n\left[ \mathbb{P}(I=i) \min\left\{1, 1 - \frac{1 - 1/e^\epsilon}{1 + (n-1)/e^\epsilon}\right\} \right] \mathrm{d}y \notag \\
&= 1 - \frac{1 - 1/e^\epsilon}{1 + (n-1)/e^\epsilon} \notag \\
&= \frac{n}{n+e^\epsilon-1}.
\end{align}
Similarly, we have
\begin{align}
A^{\mathrm{RAN}}(x) &= \int_{\mathbb{X}} f(y\mid\theta) \sum_{i=1}^n \left[\mathbb{P}(I=i) \min\left\{1, \frac{w_i(y, x)}{w_0(y, x)}\right\}\right]\mathrm{d}y \notag \\
&\ge \int_{\mathbb{X}} f(y\mid\theta) \sum_{i=1}^n\left[ \mathbb{P}(I=i) \min\left\{1, \frac{1}{e^\epsilon}\right\} \right] \mathrm{d}y \notag \\
&= \frac{1}{e^\epsilon},
\end{align}
and
\begin{align}
A_i^{\mathrm{SYS}}(x) &= \int_{\mathbb{X}} f(y\mid\theta) \min\left\{1, \frac{w_i(y, x)}{w_0(y, x)}\right\}\mathrm{d}y \notag \\
&\ge \int_{\mathbb{X}} f(y\mid\theta) \min\left\{1, \frac{1}{e^\epsilon}\right\} \mathrm{d}y \notag \\
&= \frac{1}{e^\epsilon}.
\end{align}
Thus, the proof is complete. 
\end{proof}

\subsection{The Coupling Kernel of SOMA in Section~\ref{sec:convergence}}\label{sec:coupling}

We now use coupling-based techniques to analyze the convergence rate of the algorithms. 
Our proof leverages the coupling lemma. It entails constructing two chains on a common probability space such that they eventually meet and remain coupled.

\begin{theorem}[Coupling inequality~\citep{rosenthal1995minorization}]
\label{lem:coup_ineq}
Let $\{(\Phi^{(t)}, \tilde{\Phi}^{(t)})\}_{t=0}^{\infty}$ be a pair of Markov chains, each marginally associated with an irreducible transition kernel $P$ and stationary distribution $\pi$. Suppose that
(i). if $\Phi^{(t_0)} = \tilde{\Phi}^{(t_0)}$ for some $t_0$, then $\Phi^{(t)} = \tilde{\Phi}^{(t)}$ for all $t > t_0$; and
(ii). $\tilde{\Phi}^{(0)} \sim \pi$.
Then 
\begin{equation}\label{eq:coup_ineq}
\|P^t ( \cdot \mid x) - \pi(\cdot) \|_{\mathrm{TV}} \le \mathbb{P}(\tau > t).
\end{equation}
where the meeting time is defined by $\tau = \min \{t:\ \Phi^{(t)} = \tilde{\Phi}^{(t)} \}$.
\end{theorem}

To utilize \eqref{eq:coup_ineq}, we must construct such a faithful coupling.
We create a coupling of SOMA in the two stages: (1) use a common proposed state $y$ and select a common index $I$ using a maximal coupling of Multinomial random variables~\citep{thorisson2000coupling,jacob2020couplings}; (2) accept and reject the proposed states with a common random number. 

Assume we have two chains $\Phi$ and $\tilde{\Phi}$ with corresponding states $x^{(t)}$ and $\tilde{x}^{(t)}$ at time $t$. 
Here is a coupling kernel.

1. Propose the same single offer $y\sim f(y\mid\theta)$ for both chains $\Phi$ and $\tilde{\Phi}$.

2. For chain $\Phi$, calculate weight $w_0$ and $w_i$; For chain $\tilde\Phi$, calculate weight $\tilde{w}_0$ and $\tilde{w}_i$. Denote $W=\sum_{j=1}^n w_j$ and $\tilde{W}=\sum_{j=1}^n \tilde{w}_j$.

3. Compute the component-wise minimum of selection probabilities and residuals:
\begin{equation}\label{eqn:ui}
u_i = \min \left\{\frac{w_i}{W}, \frac{\tilde{w}_i}{\tilde{W}}\right\},\quad
v_i = \frac{w_i}{W} - u_i,\quad \tilde{v}_i = \frac{\tilde{w}_i}{\tilde{W}} - u_i.
\end{equation}
In \eqref{eqn:ui},
we have $\sum_{i=1}^n u_i+\sum_{i=1}^n v_i = \sum_{i=1}^n u_i+\sum_{i=1}^n \tilde{v}_i=1$, and $\min\{v_i,\tilde{v}_i\}=0$ for all $i$.
After that, we select component $I$ for chain $\Phi$ and select component $\tilde{I}$ for chain $\tilde{\Phi}$. With probability $u_i$, we select the same component $I=\tilde{I}=i$ for both chains. If the same component is not selected, then for chain $\Phi$, we select the $i$-th component with the remaining probability $v_i$, and for chain $\tilde{\Phi}$, we select the $i$-th component with the remaining probability $\tilde{v}_i$.

As illustrated in Figure~\ref{fig:maximialcouplingofmultionomial}, we first draw a random number $U\sim\mathrm{Uniform}(0,1)$. If $U\le\sum_{j=1}^{n}u_j$, we find the unique index $i$ such that $U\in(\sum_{j=1}^{i-1}u_j,\sum_{j=1}^{i}u_j],$ and set $I=\tilde I=i$, i.e., both chains therefore choose the same component. Otherwise $U>\sum_{j=1}^{n}u_j$, we then define the residual value $U':=U-\sum_{j=1}^{n}u_j\in(0,1-\sum_{j=1}^{n}u_j]$, locating $i,\tilde{i}$ such that $U'\in(\sum_{k=1}^{i-1}v_k,\sum_{k=1}^{i}v_k]$, $U'\in(\sum_{k=1}^{\tilde{i}-1}\tilde{v}_k,\sum_{k=1}^{\tilde{i}}\tilde{v}_k]$ and set $I=i,\tilde{I}=\tilde{i}$.
\begin{figure}
    \centering
    
\begin{tikzpicture}
\draw[fill=blue!30] (0,0) rectangle (1.5,1) node[midway] {$u_1$};
\draw[fill=blue!30] (1.5,0) rectangle (2.5,1) node[midway] {$u_2$};
\draw[fill=blue!30] (2.5,0) rectangle (3.5,1) node[midway] {$\cdots$};
\draw[fill=blue!30] (3.5,0) rectangle (5,1) node[midway] {$u_n$};
\draw[fill=red!30] (5,0) rectangle (6.4,1) node[midway] {$v_1$};
\draw[fill=red!30] (6.4,0) rectangle (7.2,1) node[midway] {$v_4$};
\draw[fill=red!30] (7.2,0) rectangle (8.5,1) node[midway] {$\cdots$};
\draw[fill=red!30] (8.5,0) rectangle (10.0,1) node[midway] {$v_{n-1}$};

\draw[fill=blue!30] (0,-1.5) rectangle (1.5,-0.5) node[midway] {$u_1$};
\draw[fill=blue!30] (1.5,-1.5) rectangle (2.5,-0.5) node[midway] {$u_2$};
\draw[fill=blue!30] (2.5,-1.5) rectangle (3.5,-0.5) node[midway] {$\cdots$};
\draw[fill=blue!30] (3.5,-1.5) rectangle (5,-0.5) node[midway] {$u_n$};
\draw[fill=green!30] (5,-1.5) rectangle (5.6,-0.5) node[midway] {$\tilde{v}_2$};
\draw[fill=green!30] (5.6,-1.5) rectangle (6.9,-0.5) node[midway] {$\tilde{v}_3$};
\draw[fill=green!30] (6.9,-1.5) rectangle (8.3,-0.5) node[midway] {$\cdots$};
\draw[fill=green!30] (8.3,-1.5) rectangle (10.0,-0.5) node[midway] {$\tilde{v}_n$};

\draw[->] (6.0, 1.25) -- (6.37, 1.05) node[midway, above] {$v_2$};
\draw[->] (6.8, 1.25) -- (6.43, 1.05) node[midway, above] {$v_3$};
\draw[->] (10, 1.25) -- (10, 1.05) node[midway, above] {$v_n$};
\draw[->] (5.0, -2) -- (5.0, -1.55) node[midway, below] {$\tilde{v}_1$};
\draw (-0.5, 0.5) node {$\Phi:$};
\draw (-0.5, -1.0) node {$\tilde{\Phi}:$};
\end{tikzpicture}
    \caption{Maximal coupling of $\mathrm{Multinomial}(w)$ and $\mathrm{Multinomial}(\widetilde{w})$.}
    \label{fig:maximialcouplingofmultionomial}
\end{figure}

4. Calculate the acceptance probability
\begin{equation}
\alpha=\min\left\{ 1, \frac{W}{W+w_0-w_I} \right\},\quad 
\tilde{\alpha}=\min\left\{ 1, \frac{\tilde{W}}{\tilde{W}+\tilde{w}_0-\tilde{w}_I} \right\},
\end{equation}
and generate $\xi\sim\mathrm{Uniform}(0,1)$. 
For chain $\Phi$, set new state $x^{(t+1)}$ by replacing $x^{(t)}_I$ with $y$ if $\xi<\alpha$. For chain $\tilde{\Phi}$, set new state $\tilde{x}^{(t+1)}$ by replacing $\tilde{x}^{(t)}_{\tilde{I}}$ with $y$ if $\xi<\tilde{\alpha}$.

We now provide the following property of the coupling kernel.
\begin{lemma}\label{lem:couple_distance}
For any current state $(x,\tilde{x})$, let $(x^\prime,\tilde{x}^\prime)$ be samples from the coupled SOMA kernel $P_{\mathrm{SOMA-coupled}}$.
Under $\epsilon$-DP, we must have 
\begin{align}
&\quad P_{\mathrm{SOMA-coupled}}\left(d(x^\prime, \tilde{x}^\prime) = d(x,\tilde{x}) - 1 \mid x, \tilde{x}\right) \notag\\
&\geq  \frac{n}{n+M-1} \cdot \frac{d(x,\tilde{x})}{d(x,\tilde{x}) + e^\epsilon (n-d(x,\tilde{x}))}.
\end{align}
As a comparison, for the coupled Ran-IMwG kernel $P_{\mathrm{Ran-coupled}}$, the bound is given by
\begin{equation}
P_{\mathrm{Ran-coupled}}\left(d(x^\prime, \tilde{x}^\prime) = d(x,\tilde{x}) - 1 \mid x, \tilde{x}\right) \geq  \frac{1}{e^\epsilon} \cdot \frac{d(x,\tilde{x})}{n}.
\end{equation}
\end{lemma}

\begin{proof}[Proof of Lemma~\ref{lem:couple_distance}]
Let the set of coupled indices be $\mathcal{I} = \{i:x_i = \tilde{x}_i\}$.
In the coupling construction, if $I$ and $\tilde{I}$ are both on uncoupled positions and both chains accept the proposed component $y$, then the distance between the chains will reduce by 1.  This argument leads to 
{\allowdisplaybreaks
\begin{align*}
&\quad P_{\mathrm{SOMA-coupled}}\left(d(x^\prime, \tilde{x}^\prime) = d(x,\tilde{x}) - 1 \mid x, \tilde{x}\right)  \\
& = \int_{\mathbb{X}}  \sum_{i\in\mathcal{I}^c} \sum_{j\in\mathcal{I}^c} \mathbb{P}(I = i, \tilde{I} = j \mid x,\tilde{x}, y) \cdot \min(\alpha_i, \tilde{\alpha}_j; x, \tilde{x}, y)  f(y\mid\theta) \mathrm{d}y \\
& \geq \frac{n}{n+e^\epsilon-1} \cdot \mathbb{P}(I \in \mathcal{I}^c, \tilde{I} \in \mathcal{I}^c \mid x, \tilde{x},y)\\
&= \frac{n}{n+e^\epsilon-1} \cdot \min\left( \frac{\sum_{i\in \mathcal{I}^c} w_i}{\sum_{i=1}^n w_i}, \frac{\sum_{i\in \mathcal{I}^c} \tilde{w_i}}{\sum_{i=1}^n \tilde{w}_i} \right) \\
& \geq \frac{n}{n+e^\epsilon-1} \cdot \frac{d(x,\tilde{x})}{d(x,\tilde{x}) + e^\epsilon (n-d(x,\tilde{x}))}.
\end{align*}
}
The first inequality uses Theorem~\ref{thm:acc_bound}. 
The second lower bound is achieved when ${w_j}/{w_i} = e^\epsilon$ for all $j \in \mathcal{I}$ and all $i \in \mathcal{I}^c$. We also note that the two index sets have cardinality $d(x,\tilde{x})$ and $n - d(x, \tilde{x})$ respectively.

Similarly, for the coupled RanScan-IMwG kernel, we have
\begin{equation*}
P_{\mathrm{Ran-coupled}}\left(d(x^\prime, \tilde{x}^\prime) = d(x,\tilde{x}) - 1 \mid x, \tilde{x}\right) \ge \inf_{i\in\mathcal{I}^c}\left[\alpha_i^{\mathrm{IMwG}}(x,y)\right] \cdot \mathbb{P}(I = \tilde{I} \in \mathcal{I}^c) = \frac{1}{e^\epsilon} \cdot \frac{d(x,\tilde{x})}{n}.
\end{equation*}
\end{proof}

\begin{lemma}\label{lemma:increase}
For any current state $(x,\tilde{x})$, let $(x^\prime,\tilde{x}^\prime)$ be samples from the coupled SOMA kernel $P_{\mathrm{SOMA-coupled}}$. Under $\epsilon$-DP, we must have 
\begin{align}
&\quad P_{\mathrm{SOMA-coupled}}\left(d(x^\prime, \tilde{x}^\prime) = d(x,\tilde{x}) + 1 \mid x, \tilde{x}\right) \notag\\
& \leq \frac{(n-d)e^\epsilon}{(n-d)e^\epsilon+d} - \frac{n}{n+e^\epsilon-1}\cdot \frac{(n-d)}{(n-d)+de^\epsilon},
\end{align}
where we use $d=d(x,\tilde{x})$ for simplicity.

Specifically, for the case when $n=2$ and $d(x,\tilde{x})=1$, we have
\begin{equation}
P_{\mathrm{SOMA-coupled}}\left(d(x^\prime, \tilde{x}^\prime) = d(x,\tilde{x}) + 1 \mid x, \tilde{x}\right) \le \frac{e^\epsilon(e^\epsilon-1)}{(e^\epsilon+1)^2}.
\end{equation}
\end{lemma}

\begin{proof}\label{proof:increase}
We start by considering the general case. Using the same notation as in Section~\ref{sec:coupling}, suppose that the first $d$ elements of $x$ and $\tilde{x}$ differ, i.e., $x_i\neq\tilde{x}_i$ for $i=1,\ldots,d$ and $x_j=\tilde{x}_j$ for $j=d+1,\ldots,n$. Without loss of generality, assume that $\sum_{i=1}^d v_i \ge \sum_{i=1}^d \tilde{v}_i$. Now consider the following cases:

\noindent \textbf{Case (a)}. $I = \tilde{I} \in \{1,\ldots, d\}$. This case does not increase the coupling distance.

\noindent \textbf{Case (b)}. $I = \tilde{I} \in \{d+1,\ldots, n\}$. This case can increase the coupling distance only when one chain accepts the update while the other rejects it. The probability of this occurring is given by $|\alpha_I - \tilde{\alpha}_{\tilde{I}}|$.

\noindent \textbf{Case (c)}. $I \in \{1,\ldots, d\}, \tilde{I} \in \{1,\ldots, d\}, I \neq \tilde{I}$. This case does not increase the coupling distance.

\noindent \textbf{Case (d)}. $I \in \{1,\ldots, d\}, \tilde{I} \in \{d+1,\ldots, n\}$. This case increases the coupling distance only if the update to chain $\Phi$ is rejected while the update to chain $\tilde{\Phi}$ is accepted. The probability is given by $\max\{\tilde{\alpha}_{\tilde{I}} - \alpha_I, 0\}$

\noindent \textbf{Case (e)}. $I \in \{d+1,\ldots, n\}, \tilde{I} \in \{1,\ldots, d\}$. This case cannot occur, due to our assumption that $\sum_{i=1}^d v_i \ge \sum_{i=1}^d \tilde{v}_i$.

\noindent \textbf{Case (f)}. $I \in \{d+1,\ldots, n\}, \tilde{I} \in \{d+1,\ldots, n\}, I \neq \tilde{I}$. In this case, accepting an update in either chain increases the coupling distance. The corresponding probability is $\max\{\alpha_I, \tilde{\alpha}_{\tilde{I}}\}$.

\begin{figure}
    \centering
    
\begin{tikzpicture}
\draw[fill=blue!30] (0,0) rectangle (1.0,1) node[midway] {$u_1$};
\draw[fill=blue!30] (1.0,0) rectangle (2.0,1) node[midway] {$\cdots$};
\draw[fill=blue!30] (2.0,0) rectangle (3.0,1) node[midway] {$u_d$};
\draw[fill=blue!30] (3.0,0) rectangle (4.0,1) node[midway] {$u_{d+1}$};
\draw[fill=blue!30] (4.0,0) rectangle (5.0,1) node[midway] {$\cdots$};
\draw[fill=blue!30] (5.0,0) rectangle (6.0,1) node[midway] {$u_n$};
\draw[fill=red!30] (6.0,0) rectangle (9.0,1) node[midway] {$v_1,\ldots,v_d$};
\draw[fill=red!30] (9.0,0) rectangle (11.0,1) node[midway] {$v_{d+1},\ldots,v_n$};

\draw[fill=blue!30] (0,-1.5) rectangle (1.0,-0.5) node[midway] {$u_1$};
\draw[fill=blue!30] (1.0,-1.5) rectangle (2.0,-0.5) node[midway] {$\cdots$};
\draw[fill=blue!30] (2.0,-1.5) rectangle (3.0,-0.5) node[midway] {$u_d$};
\draw[fill=blue!30] (3.0,-1.5) rectangle (4.0,-0.5) node[midway] {$u_{d+1}$};
\draw[fill=blue!30] (4.0,-1.5) rectangle (5.0,-0.5) node[midway] {$\cdots$};
\draw[fill=blue!30] (5.0,-1.5) rectangle (6.0,-0.5) node[midway] {$u_n$};
\draw[fill=green!30] (6.0,-1.5) rectangle (8.0,-0.5) node[midway] {$\tilde{v}_1,\ldots,\tilde{v}_d$};
\draw[fill=green!30] (8.0,-1.5) rectangle (11.0,-0.5) node[midway] {$\tilde{v}_{d+1},\ldots,\tilde{v}_n$};

\draw[decorate, decoration = {brace,mirror,amplitude=5pt}] (0.05,-1.7) -- (2.95,-1.7) node[midway,above=-25pt] {Case (a)};
\draw[decorate, decoration = {brace,mirror,amplitude=5pt}] (3.05,-1.7) -- (5.95,-1.7) node[midway,above=-25pt] {Case (b)};
\draw[decorate, decoration = {brace,mirror,amplitude=5pt}] (6.05,-1.7) -- (7.95,-1.7) node[midway,above=-25pt] {Case (c)};
\draw[decorate, decoration = {brace,mirror,amplitude=5pt}] (8.05,-1.7) -- (8.95,-1.7) node[midway,above=-25pt] {Case (d)};
\draw[decorate, decoration = {brace,mirror,amplitude=5pt}] (9.05,-1.7) -- (10.95,-1.7) node[midway,above=-25pt] {Case (f)};

\draw (-0.5, 0.5) node {$\Phi:$};
\draw (-0.5, -1.0) node {$\tilde{\Phi}:$};
\end{tikzpicture}
    \caption{Possible subscript choices for coupled chains.}
    \label{fig:subscripts_choices}
\end{figure}

Let $\mathcal{I}=\{d+1,\ldots,n\}$ be the coupled index, then we have:
{\allowdisplaybreaks
\begin{align}
&\quad P_{\mathrm{SOMA-coupled}}\left(d(x^\prime, \tilde{x}^\prime) = d(x,\tilde{x}) + 1 \mid x, \tilde{x}\right)  \notag\\
& = \int_{\mathbb{X}} \sum_{i\in\mathcal{I}} \mathbb{P}(I=\tilde{I}=i \mid x, \tilde{x}, y) \cdot (|\alpha_i - \tilde{\alpha}_j|; x, \tilde{x}, y) f(y\mid\theta) \mathrm{d}y \quad (\text{Case (b)}) \notag \\
& \quad + \int_{\mathbb{X}} \sum_{i\in\mathcal{I}^c} \sum_{j\in\mathcal{I}} \mathbb{P}(I=i, \tilde{I}=j \mid x, \tilde{x}, y) \cdot (\max\{\tilde{\alpha}_j - \alpha_i, 0\}; x, \tilde{x}, y) f(y\mid\theta) \mathrm{d}y \quad (\text{Case (d)}) \notag \\
& \quad + \int_{\mathbb{X}} \sum_{i\in\mathcal{I}} \sum_{j\in\mathcal{I}} \mathbb{P}(I=i, \tilde{I}=j, i \neq j \mid x, \tilde{x}, y) \cdot (\max\{\alpha_i, \tilde{\alpha}_j\}; x, \tilde{x}, y) f(y\mid\theta) \mathrm{d}y. \quad (\text{Case (f)})\label{eq:d_plus_2}
\end{align}
}
Based on the fact that $\max\{\tilde{\alpha}_j - \alpha_i, 0\} \le \max\{\alpha_i, \tilde{\alpha}_j\} \le 1$, we have the following control:
{\allowdisplaybreaks
\begin{align}
&\quad P_{\mathrm{SOMA-coupled}}\left(d(x^\prime, \tilde{x}^\prime) = d(x,\tilde{x}) + 1 \mid x, \tilde{x}\right)  \notag\\
& \le \int_{\mathbb{X}} \sum_{i=d+1}^n u_i \cdot (|\alpha_i - \tilde{\alpha}_i|; x, \tilde{x}, y) f(y\mid\theta) \mathrm{d}y + \int_{\mathbb{X}} \sum_{i=d+1}^n \tilde{v}_i \cdot 1 \cdot f(y\mid\theta)\mathrm{d}y \notag \\
& \le \int_{\mathbb{X}} \left[ \left(1-\frac{n}{n+e^\epsilon-1}\right) \sum_{i=d+1}^n u_i + \sum_{i=d+1}^n \tilde{v}_i \right] f(y\mid\theta)\mathrm{d}y \notag \\
& = \int_{\mathbb{X}} \sum_{i=d+1}^n (\tilde{v}_i+u_i) f(y\mid\theta)\mathrm{d}y -\frac{n}{n+e^\epsilon-1} \int_{\mathbb{X}} \sum_{i=d+1}^n u_i f(y\mid\theta)\mathrm{d}y \notag\\
& = \int_{\mathbb{X}} \frac{\sum_{i=d+1}^n \tilde{w}_i}{\tilde{W}} f(y\mid\theta)\mathrm{d}y -\frac{n}{n+e^\epsilon-1} \int_{\mathbb{X}} \sum_{i=d+1}^n u_i f(y\mid\theta)\mathrm{d}y
.\label{eq:d_plus_3}
\end{align}
}
Now, under the $\epsilon$-DP, $\frac{\sum_{i=d+1}^n w_i}{W}$ has the maximum value $\frac{(n-d)e^\epsilon}{(n-d)e^\epsilon+d}$ when $w_1=w_2=\ldots=w_d$ and $e^\epsilon w_1=w_{d+1}=w_{d+2}=\ldots=w_n$. As for $\frac{\sum_{i=d+1}^n w_i}{W}$, it has the minimum value $\frac{(n-d)}{(n-d)+de^\epsilon}$ when $\tilde{w}_1=\tilde{w}_2=\ldots=\tilde{w}_d$ and $\tilde{w}_1/e^\epsilon=\tilde{w}_{d+1}=\tilde{w}_{d+2}=\ldots=\tilde{w}_n$. In this case, $\sum_{i=d+1}^n u_i=\sum_{i=d+1}^n \min\left\{ \frac{w_i}{W},\frac{\tilde{w_i}}{\tilde{W}} \right\}$ has the minimun value $\frac{(n-d)}{(n-d)+de^\epsilon}$. It then gives
\begin{equation*}
P_{\mathrm{SOMA-coupled}}\left(d(x^\prime, \tilde{x}^\prime)  d(x,\tilde{x}) + 1 \mid x, \tilde{x}\right) 
\le \frac{(n-d)e^\epsilon}{(n-d)e^\epsilon+d} - \frac{n}{n+e^\epsilon-1}\cdot \frac{(n-d)}{(n-d)+de^\epsilon}.
\end{equation*}

For the case $n=2$ and $d(x,\tilde{x})=1$, Case (f) cannot occur and we have
{\allowdisplaybreaks
\begin{align}
&\quad P_{\mathrm{SOMA-coupled}}\left(d(x^\prime, \tilde{x}^\prime) = d(x,\tilde{x}) + 1 \mid x, \tilde{x}\right)  \notag\\
& \le \int_{\mathbb{X}} u_2 \cdot |\alpha_2 - \tilde{\alpha}_2| \cdot f(y\mid\theta) \mathrm{d}y + \int_{\mathbb{X}} \tilde{v}_2 \cdot \max\{ \tilde{\alpha}_2 - \alpha_1, 0 \} f(y\mid\theta) \mathrm{d}y \notag\\
& \le \left( 1 - \frac{2}{e^\epsilon+1} \right) \int_{\mathbb{X}} u_2 \cdot f(y\mid\theta) \mathrm{d}y + \left( 1 - \frac{2}{e^\epsilon+1} \right) \int_{\mathbb{X}} \tilde{v}_2 \cdot f(y\mid\theta) \mathrm{d}y \notag\\
& = \left( 1 - \frac{2}{e^\epsilon+1} \right) \int_{\mathbb{X}} \frac{\tilde{w}_2}{\tilde{W}} \cdot f(y\mid\theta) \mathrm{d}y \notag\\
& \le \left( 1 - \frac{2}{e^\epsilon+1} \right) \frac{e^\epsilon}{e^\epsilon+1} \notag\\
& = \frac{e^\epsilon(e^\epsilon-1)}{(e^\epsilon+1)^2}
.\label{eq:d_plus_n_eq_2}
\end{align}
}
\end{proof}

\begin{proof}[Proof of Theorem~\ref{thm:convergence_rate}]
Define $d_t = d(x^{(t)}, \tilde{x}^{(t)})$ as the number of differing components at time $t$. The sequence $\mathcal{D} := \{d_t\}_{t=0}^\infty$ forms a Markov chain with discrete state space $\{0, 1, \dots, n\}$. The meeting time of the two chains is 
$\tau = \min \{t : d_t = 0\}.$
For the chain $\mathcal{D}$, let $P_d = (p_{i,j})_{n \times n}$ denote its transition matrix, where $p_{i,j} := \mathbb{P}(d_{t+1} = j \mid d_t = i)$. For our proposed coupling method in Algorithm 2, the matrix $P_d$ for the case $n = 2$ is given by:
\begin{equation}
P_d = \begin{pmatrix}
1 & 0 & 0 \\
p_{10} & p_{11} & p_{12} \\
0 & p_{21} & p_{22}
\end{pmatrix}.
\end{equation}
Our strategy is to find upper bound for $p_{12}$ and lower bounds for $p_{10}$ and $p_{21}$.These bounds lead to a kernel $Q_d$ that stochastically dominate $P_d$:
\begin{equation}
Q_d := \begin{pmatrix}
1 & 0 & 0 \\
q_{10} \le p_{10} & q_{11} & q_{12} \ge p_{12} \\
0 & q_{21} \le p_{21} & q_{22} \ge p_{22}
\end{pmatrix}.
\end{equation}The absorbing time of $Q$ is an upper bound on the absorbing time of $P$, which is the meeting time of our coupled Markov chains.

\textbf{Convergence rate of SOMA.}
Next, we find the lower bounds for $p_{10}$ and $p_{21}$. 
Based on Lemma~\ref{lem:couple_distance} with $n=2$, we have
\begin{equation*}
p_{10} \ge \frac{2}{2+e^\epsilon-1}\frac{1}{1+e^\epsilon} = \frac{2}{(e^\epsilon+1)^2},
\end{equation*}
and
\begin{equation*}
p_{21} \ge \frac{2}{2+e^\epsilon-1}\frac{2}{2} = \frac{2}{e^\epsilon+1}.
\end{equation*}

For the upper bound of $p_{12}$, Lemma~\ref{lemma:increase} gives $p_{12} \le \frac{e^\epsilon(e^\epsilon-1)}{(e^\epsilon+1)^2}$.

We define the matrix $Q_d$ with row sums equal to 1, representing the stochastic dominance of $P_d$:
\begin{equation}
Q_d := \begin{pmatrix}
1 & 0 & 0 \\
\frac{2}{(e^\epsilon+1)^2} & \frac{3e^\epsilon-1}{(e^\epsilon+1)^2} & \frac{e^\epsilon(e^\epsilon-1)}{(e^\epsilon+1)^2} \\
0 & \frac{2}{e^\epsilon+1} & \frac{e^\epsilon-1}{e^\epsilon+1}
\end{pmatrix}.
\end{equation}
For $Q_d$, the largest eigenvalue is $\lambda_1 = 1$ and the second largest eigenvalue is:
\begin{equation}
\lambda_2 = \frac{e^{2\epsilon} + 3e^\epsilon - 2 + \sqrt{e^{4\epsilon} + 2e^{3\epsilon} + 9e^{2\epsilon} - 8e^{\epsilon}}}{2(1 + e^\epsilon)^2}.
\end{equation}
Thus, we have $\mathbb{P}(\tau > t) \le C \lambda_2^t$, providing the upper bound for the convergence rate of the SOMA algorithm.

\textbf{Convergence rate of RanScan-IMwG.} Let
\begin{equation}
P_d^{\mathrm{RAN}} = \begin{pmatrix}
1 & 0 & 0 \\
p_{10}^{\mathrm{RAN}} & p_{11}^{\mathrm{RAN}} & p_{12}^{\mathrm{RAN}} \\
0 & p_{21}^{\mathrm{RAN}} & p_{22}^{\mathrm{RAN}}
\end{pmatrix}.
\end{equation}
denote its transition matrix. 
Based on Lemma~\ref{lem:couple_distance} with $n=2$, we have
\begin{equation*}
p_{10}^{\mathrm{RAN}}\ge\frac{1}{2e^\epsilon}, \quad
p_{21}^{\mathrm{RAN}}\ge\frac{1}{e^\epsilon},
\end{equation*}
and
\begin{equation*}
p_{12}^{\mathrm{RAN}}\le \frac12 \mathbb{P}(\alpha < \xi < \tilde{\alpha} \ \mathrm{or}\ \tilde{\alpha} < \xi < \alpha) \le \frac12 \left( 1 - \frac{1}{e^\epsilon}\right).
\end{equation*}
Thus, we define the corresponding stochastic dominance matrix as
\begin{equation}
Q_d^{\mathrm{RAN}} := \begin{pmatrix}
1 & 0 & 0 \\
\frac{1}{2e^\epsilon} & \frac{1}{2} & \frac12 - \frac{1}{2e^\epsilon} \\
0 & \frac{1}{e^\epsilon} & 1-\frac{1}{e^\epsilon}
\end{pmatrix},
\end{equation}
The second largest eigenvalue of $Q_d^{\mathrm{RAN}}$ is then
\begin{equation}
\lambda_2^{\mathrm{RAN}} = \frac{3e^\epsilon-2+\sqrt{e^{2\epsilon}+4e^\epsilon-4}}{4e^\epsilon}.
\end{equation}

\textbf{Convergence rate of SysScan-IMwG.} Since this algorithm updates components in a fixed order, the indices chosen for updating in the two coupled chains are always identical. For $n=2$, we view every two consecutive iterations as one overall update and denote the new time index by $h=2t$. Let
\begin{equation}
P_h^{\mathrm{SYS}} = \begin{pmatrix}
1 & 0 & 0 \\
p_{10}^{\mathrm{SYS}} & p_{11}^{\mathrm{SYS}} & p_{12}^{\mathrm{SYS}} \\
p_{20}^{\mathrm{SYS}} & p_{21}^{\mathrm{SYS}} & p_{22}^{\mathrm{SYS}}
\end{pmatrix}
\end{equation}
be the transition matrix for the chain $\{d_h\}_{h=0}^\infty$. In this case, we have the following inequalities:
\begin{align*}
p_{10}^{\mathrm{SYS}} &\ge \min\{\alpha,\tilde{\alpha}\} \min\{\alpha,\tilde{\alpha}\} \ge \frac{1}{e^{2\epsilon}}, \\
p_{20}^{\mathrm{SYS}} &\ge \min\{\alpha,\tilde{\alpha}\} \min\{\alpha,\tilde{\alpha}\} \ge \frac{1}{e^{2\epsilon}},
\end{align*}
and
\begin{align*}
p_{12}^{\mathrm{SYS}} \le \min\{1-\alpha,1-\tilde{\alpha}\} \min\{1-\alpha,1-\tilde{\alpha}\} \ge \left(1 - \frac{1}{e^\epsilon}\right)^2, \\
p_{22}^{\mathrm{SYS}} \le \min\{1-\alpha,1-\tilde{\alpha}\} \min\{1-\alpha,1-\tilde{\alpha}\} \ge \left(1 - \frac{1}{e^\epsilon}\right)^2.
\end{align*}
Accordingly, the stochastic dominance matrix is defined as
\begin{equation}
Q_h^{\mathrm{SYS}} := \begin{pmatrix}
1 & 0 & 0 \\
\frac{1}{e^{2\epsilon}} & 1-\frac{1}{e^{2\epsilon}}-\left(1 - \frac{1}{e^\epsilon}\right)^2 & \left(1 - \frac{1}{e^\epsilon}\right)^2 \\
\frac{1}{e^{2\epsilon}} & 1-\frac{1}{e^{2\epsilon}}-\left(1 - \frac{1}{e^\epsilon}\right)^2 & \left(1 - \frac{1}{e^\epsilon}\right)^2
\end{pmatrix}.
\end{equation}
The second largest eigenvalue of $Q_h^{\mathrm{SYS}}$ is
\begin{equation*}
\lambda_2^{\mathrm{SYS}} = \frac{(e^\epsilon+1)(e^\epsilon-1)}{e^{2\epsilon}}.
\end{equation*}
Since the updates are performed every two iterations, we have
$\mathbb{P}(\tau > 2t) \le C \left(\lambda_2^{\mathrm{SYS}}\right)^t$,
or equivalently,
$\mathbb{P}(\tau > t) \le C \left(\sqrt{\lambda_2^{\mathrm{SYS}}}\right)^t$.
Thus, the convergence rate satisfies
\begin{equation}
r_{\mathrm{SYS}} \le \frac{\sqrt{(e^\epsilon+1)(e^\epsilon-1)}}{e^\epsilon}.
\end{equation}
This completes the proof.
\end{proof}

\section{EXPERIMENTAL DETAILS AND FURTHER RESULTS}\label{appendix:exp_detail}

\subsection{Data Imputation from Perturbed Histograms}\label{sec:exp_hist}
We apply SOMA to data imputation from perturbed histograms~\citep{wasserman2010statistical}. 
This is a specific example of \Cref{example:perturbedhistogram}.
Our experiments use proposal $f(x) = 1$ for $x \in \mathbb{X} = (0,1)$ and privacy level $\epsilon = 5$, and $m = 10$ equidistant bins. We vary the number of observations from $n = 2$ to $n = 60$. 

\begin{figure*}[htbp]
\hspace{0.0\textwidth}
\begin{subfigure}[t]{0.49\textwidth}
    \vspace{0pt}
    \includegraphics[width = 1\textwidth]{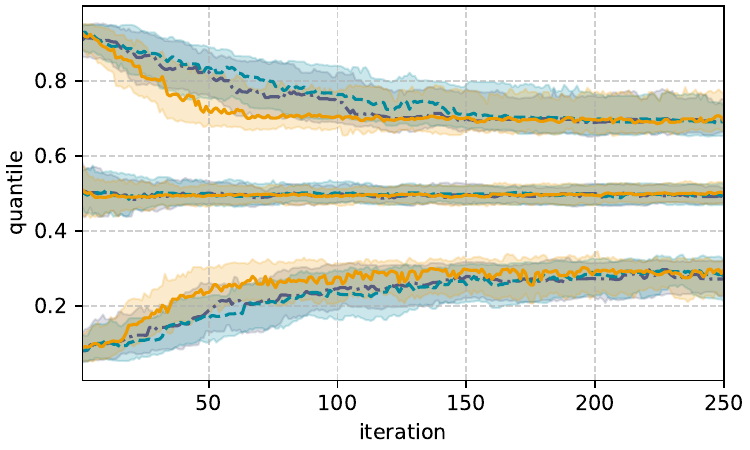}
\end{subfigure}
\hspace{-0.50\textwidth}
\begin{subfigure}[t]{0.02\textwidth}
    \vspace{2pt}
    \textbf{A}
\end{subfigure}
\hspace{0.47\textwidth}
\begin{subfigure}[t]{0.49\textwidth}
    \vspace{0pt}
    \includegraphics[width = 1\textwidth]{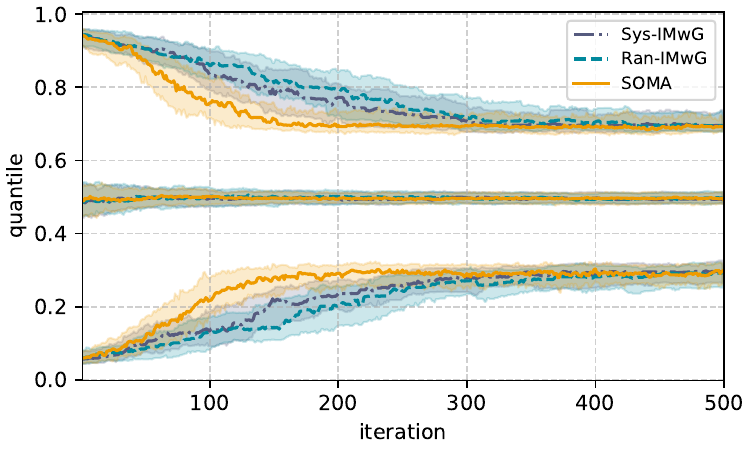}
\end{subfigure}
\hspace{-0.50\textwidth}
\begin{subfigure}[t]{0.02\textwidth}
    \vspace{2pt}
    \textbf{B}  
\end{subfigure}
\caption{
Quantiles of the states.
\textbf{A.} Components $n=20$.
\textbf{B.} Components $n=50$.
The solid lines represent the median, and the shaded regions indicate the quartiles over 100 repeated runs.
}
\label{fig:phist_quantile}
\end{figure*}
To compare the three samplers, we first compare the traceplot of the quartiles of $X$ given $S_\mathrm{dp}$ for $n=20$ and $n=50$. These traceplots suggest that SOMA converges more rapidly.
\begin{figure*}[ht]
\hspace{0.0\textwidth}
\begin{subfigure}[t]{0.49\textwidth}
    \vspace{0pt}
    \includegraphics[width = 1\textwidth]{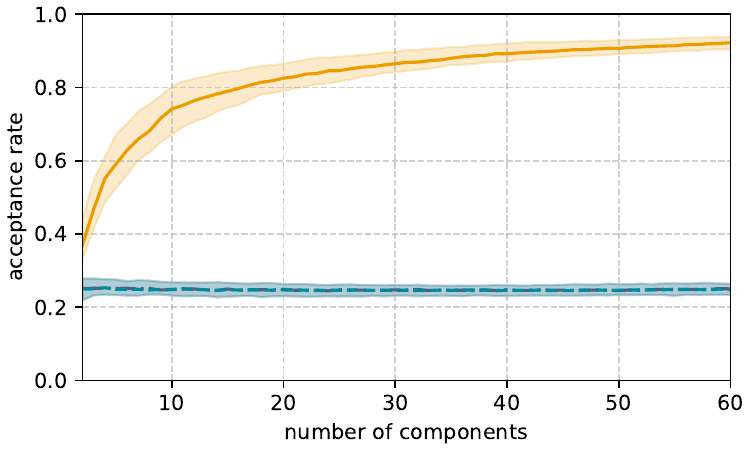}
\end{subfigure}
\hspace{-0.50\textwidth}
\begin{subfigure}[t]{0.02\textwidth}
    \vspace{2pt}
    \textbf{A}
\end{subfigure}
\hspace{0.47\textwidth}
\begin{subfigure}[t]{0.49\textwidth}
    \vspace{0pt}
    \includegraphics[width = 1\textwidth]{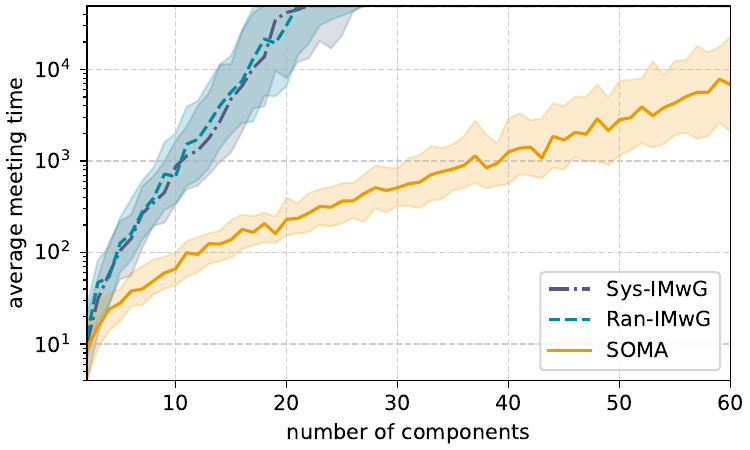}
\end{subfigure}
\hspace{-0.50\textwidth}
\begin{subfigure}[t]{0.02\textwidth}
    \vspace{2pt}
    \textbf{B}  
\end{subfigure}
\caption{
\textbf{A.} Average acceptance rate under $T=50,000$ iterations with different $n$.
\textbf{B.} Average meeting times for the coupled chains with different $n$.
The solid lines represent the median, and the shaded regions indicate the quartiles over 100 repeated runs.
}
\label{fig:phist_diff_n}
\end{figure*}

For a quantitative comparison, we examine their average acceptance rates and meeting times in Figure~\ref{fig:phist_diff_n}.
The advantage of SOMA over the other two methods is more pronounced as the number of observations grows. 
SOMA accepts proposed points more often and can be coupled faster.
In particular, to couple two chains at $n = 20$, SOMA requires only about 1\% of the iterations needed by the other two methods.

\begin{figure*}[t]
\hspace{0.0\textwidth}
\begin{subfigure}[t]{0.49\textwidth}
    \vspace{0pt}
    \includegraphics[width = 1\textwidth]{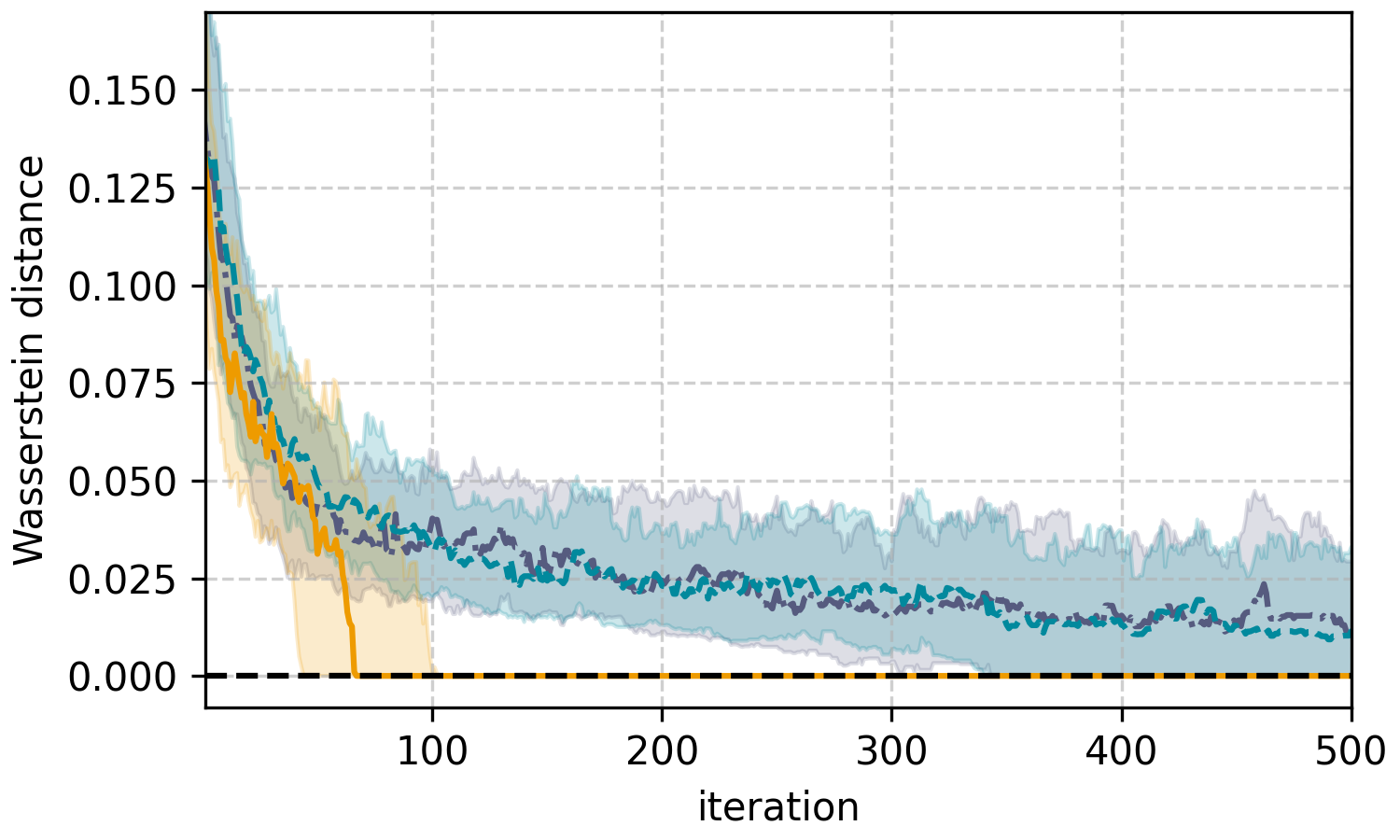}
\end{subfigure}
\hspace{-0.50\textwidth}
\begin{subfigure}[t]{0.02\textwidth}
    \vspace{2pt}
    \textbf{A}
\end{subfigure}
\hspace{0.47\textwidth}
\begin{subfigure}[t]{0.49\textwidth}
    \vspace{0pt}
    \includegraphics[width = 1\textwidth]{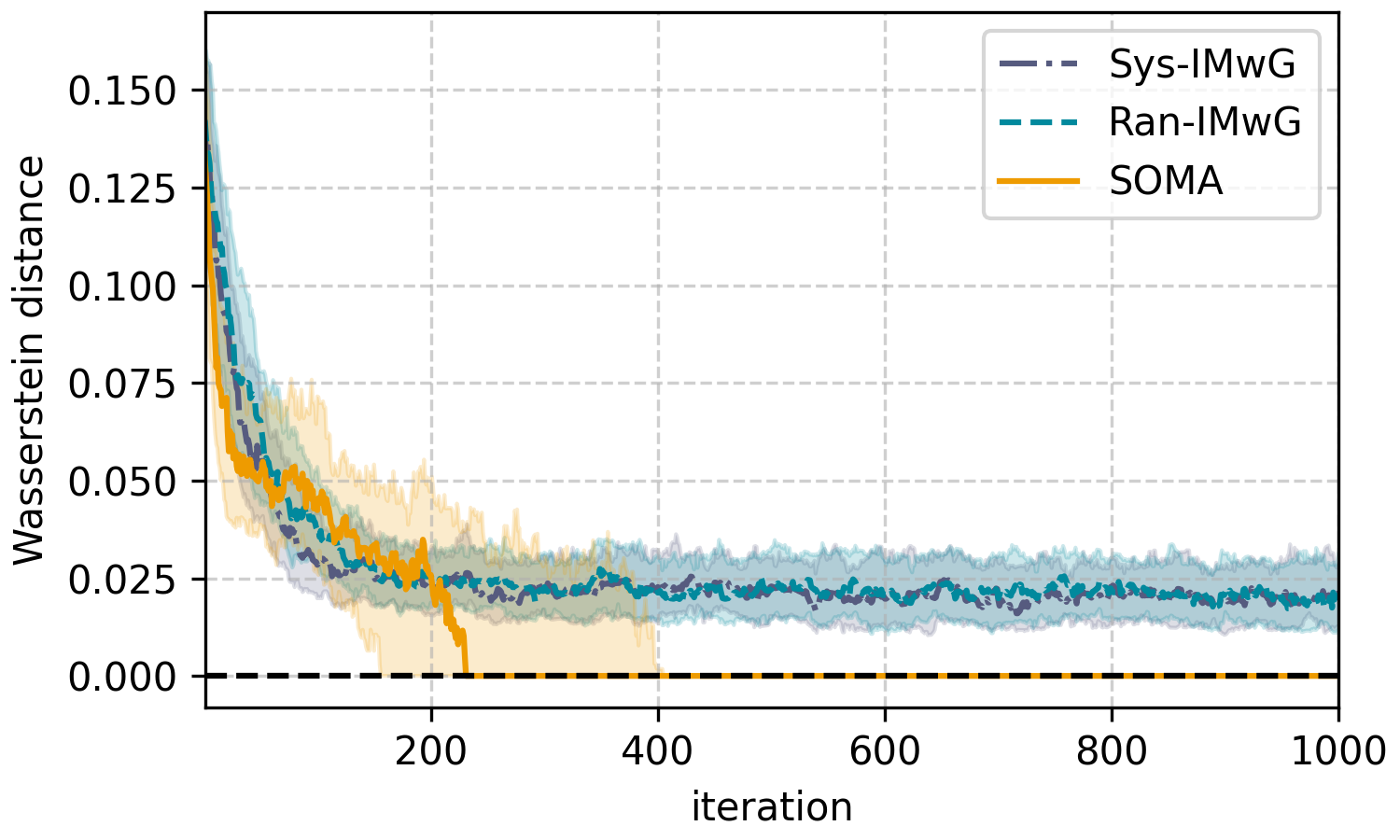}
\end{subfigure}
\hspace{-0.50\textwidth}
\begin{subfigure}[t]{0.02\textwidth}
    \vspace{2pt}
    \textbf{B}  
\end{subfigure}
\caption{
Wasserstein distance between states $X^{(t)}$ and $\tilde{X}^{(t)}$.
The number of components is $n = 10$ in \textbf{A} and $n = 20$ in \textbf{B}.
The solid lines represent the median, and the shaded regions indicate the quartiles over 100 repeated runs.
}
\label{fig:phist_was_dist}
\end{figure*}

In Figure~\ref{fig:phist_diff_n}-B, we insist that the two chains meet when $X = \tilde{X}$ on all components. 
It is possible that the two states $X^{(t)}$ and $\tilde{X}^{(t)}$ have already been `coupled' up to shuffling the indices. 
To test this, we examine the Wasserstein distance~\citep{cuturi2013sinkhorn} between empirical measures defined by the states, i.e., $W_2(X^{(t)}, \tilde{X}^{(t)})$ at $n = 10$ and $n = 20$ in Figure~\ref{fig:phist_was_dist}.
As shown in all three coupling kernels can shrink the Wasserstein distance rapidly during an initial phase. 
However, the coupling kernel of SOMA can keep contracting the two chains closer until they meet while the other two methods cannot further contract $X^{(t)}$ and $\tilde{X}^{(t)}$.
This experiment confirms that SOMA converges faster than the other two methods, as the coupling kernel of SOMA can make the chains meet faster with respect to a distance that is insensitive to component indices.

\subsection{Additional Details on Section~\ref{sec:exp_syn}}\label{appendix:details_syn}

\textbf{Estimating convergence rate with coupling time.}
One can estimate convergence rates $r$ by running independent coupled chains. 
With an empirical survival function of the coupling time $\mathbb{P}(\tau > t)$, 
if we regress $\log \mathbb{P}(\tau > t)$ on $t$, the slope gives an estimate of $\log(r)$.
Figure~\ref{fig:syn_lm} shows the fitted regression lines based on running 100 coupled chains independently with $a_0=b_0=10$ and $\epsilon=20$.
In this case, SOMA converges faster than the other two methods.
Notably, the empirical convergence rate suggests that the theoretical upper bound derived in Theorem~\ref{thm:convergence_rate} has been conservative
and that it is possible to derive sharper bounds for specific target densities.

\begin{figure}[htbp]
\centering
\includegraphics[width=0.5\textwidth]{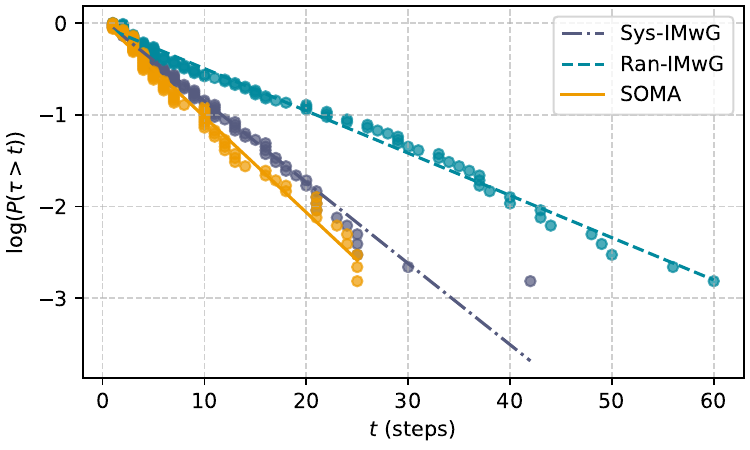}
\caption{Log survival function of the coupling time $\log \mathbb{P}(\tau > t)$ with $\epsilon = 20$.}
\label{fig:syn_lm}
\end{figure}

\textbf{Convergence rate under strong posterior correlation.}
Using coupled chains, we can estimate the convergence rates using the empirical distribution of the coupling time. This enables us to investigate the impact of posterior correlation $\rho_{\pi}(X_1,X_2)$ on convergence speed, illustrated in Figure~\ref{fig:syn_convergence_compare}.
Higher posterior correlation typically indicates a more challenging sampling scenario, 
and indeed all three samplers mix more slowly under stronger posterior correlation. 
In both simulation setups (with or without $\epsilon$-DP), SOMA is the most efficient in the high correlation scenario ($\rho < -0.4$) and Sys-IMwG is the most efficient when $\rho$ is small.

\begin{figure*}[htbp]
\hspace{0.0\textwidth}
\begin{subfigure}[t]{0.49\textwidth}
    \vspace{0pt}
    \includegraphics[width = 1\textwidth]{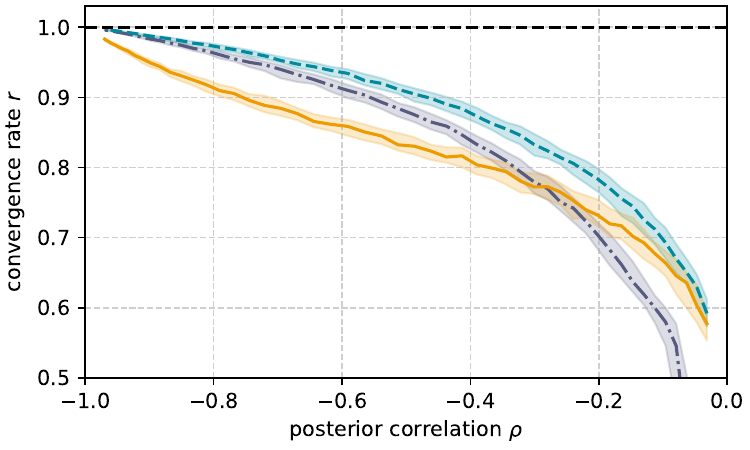}
\end{subfigure}
\hspace{-0.50\textwidth}
\begin{subfigure}[t]{0.02\textwidth}
    \vspace{2pt}
    \textbf{A}
\end{subfigure}
\hspace{0.47\textwidth}
\begin{subfigure}[t]{0.49\textwidth}
    \vspace{0pt}
    \includegraphics[width = 1\textwidth]{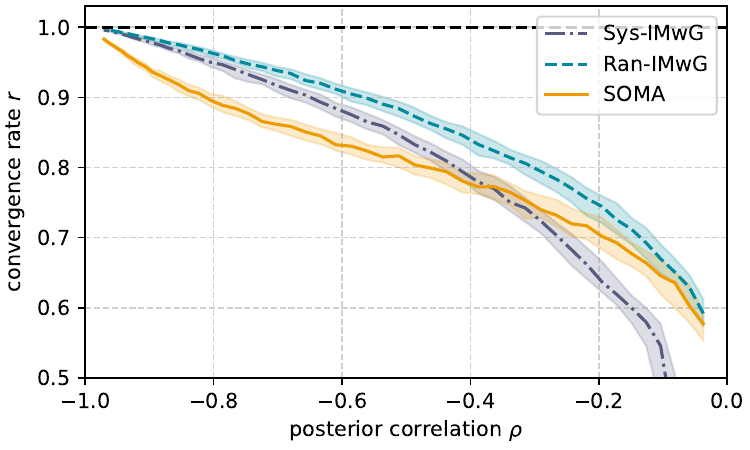}
\end{subfigure}
\hspace{-0.50\textwidth}
\begin{subfigure}[t]{0.02\textwidth}
    \vspace{2pt}
    \textbf{B}  
\end{subfigure}
\caption{
Convergence rates with different posterior correlations.
\textbf{A.} Using $\mathrm{Beta}(10, 10)$ priors on $X_i.$ This setting satisfies $\epsilon$-DP.
\textbf{B.} Using $\mathrm{Exp}(1)$ priors. The $\epsilon$-DP is violated.}
\label{fig:syn_convergence_compare}
\end{figure*}

\subsection{Details of Bayesian Linear Regression Given Privatized Data}\label{appendix:details_pbl}

We model the predictors with $x_{0,i} \sim\mathcal{N}_p(m, \Sigma)$, where $\Sigma=I_n$ and $m=(0.9,\ -1.17)$. Let $x=(\textbf{1}_{n\times 1}, x_0)$ represent the design matrix with the intercept, the outcomes $y$ satisfy $y\mid x_0\sim\mathcal{N}_n(x\beta, \sigma^2 I_n)$. 
We use conjugate priors on the coefficients $\beta$ and regression noise $\sigma$, i.e. the normal-inverse gamma (NIG) prior $(\beta,\sigma^2)\sim\mathrm{NIG}(\mu_0, \Lambda_0, a_0, b_0)$, which means $\sigma^2\sim\mathrm{InverseGamma}(a_0, b_0)$ and $\beta\mid\sigma^2\sim\mathcal{N}(\mu_0,\sigma^2 \Lambda_0^{-1})$. In this case, the posterior distribution is also a normal-inverse gamma distribution $\beta,\sigma^2 \mid x, y \sim\mathrm{NIG}(\mu_n, \Lambda_n, a_n, b_n)$
\begin{align*}
\mu_n&= (x^\top X + \Lambda_0)^{-1}(x^\top y + \mu_0^\top \Lambda_0) \\
\Lambda_n &= x^\top X + \Lambda_0 \\
a_n &= a_0 + \frac{n}{2} \\
b_n &= b_0 + \frac12(y^\top y + \mu_0^\top \Lambda_0 \mu_0 - \mu_n^\top \Lambda_n \mu_n).
\end{align*}
We set $a_0 = 10$, $b_0 = 10$, $\mu_0 = \mathbf{0}_{p+1}$, and $\Lambda_0 = 0.5 I_{p+1}$ with $p=2$ for our experiments.
The ground truth parameters $\theta^*=(\beta^*,\sigma^{2*})$ are denoted as
\begin{equation*}
\beta^* = (-1.79,\ -2.89,\ -0.66)^\top,\quad \sigma^{2*}=1.13.
\end{equation*}
For $\epsilon=30$, we simulate the privatized summary statistics $s_{\mathrm{dp}}$ from the model using the ground truth parameters $\theta^*$, resulting in
\begin{equation*}
s_{\mathrm{dp}} = \left(
\left(\begin{array}{c}
-0.5359 \\
-0.1340 \\
0.0851 \\
\end{array}\right),\
0.4365,\
\left(\begin{array}{ccc}
1.0000 & 0.1517 & -0.1787 \\
0.1517 & 0.0483 & -0.0280 \\
-0.1787 & -0.0280 & 0.0600 \\
\end{array}\right)
\right),
\end{equation*}
its corresponding vector form after removing duplicate entries is
\begin{equation*}
s_{\mathrm{dp}}^{\mathrm{vec}} = (-0.5359,\ -0.1340,\  0.0851,\  0.4365,\  0.1517,\ -0.1787,\  0.0483,\ -0.0280,\ 0.0600).
\end{equation*}
For $\epsilon=3$, the corresponding vector form of $s_{\mathrm{dp}}$ is given by
\begin{equation*}
s_{\mathrm{dp}}^{\mathrm{vec}} = (-0.5120,\ -0.1326,\  0.1554,\  0.4520,\  0.0471,\ -0.1216,\  0.0803,\ -0.0676,\ 0.0019).
\end{equation*}

\begin{table}[htbp]
\caption{Coupling time, convergence rate, acceptance rate for $n = 10$ and $\epsilon = 3$ and $\epsilon = 30.$ All values are presented as mean with quartiles.}
\label{tab:pbl_coupling1}
\centering
  \begin{tabular}{c|ccc}
    Method & Coupling Time & Convergence Rate & Acceptance Rate (\%) \\ 
    \hline
    \multicolumn{4}{c}{$n=10,\ \epsilon=3$} \\
    \hline
    Ran-IMwG & $156.04\ (46, 216)$ & $0.9937\ (0.9933, 0.9941)$ & $94.80\ (94.52, 95.11)$ \\ 
    Sys-IMwG & $\textbf{99.60}\ (\textbf{28}, \textbf{138})$ & $\textbf{0.9904}\ (\textbf{0.9895}, \textbf{0.9911})$ & $94.81\ (94.55, 95.09)$ \\ 
    SOMA & $152.98\ (42, 211)$ & $ 0.9937\ (0.9932, 0.9942)$ & $\textbf{99.49}\ (\textbf{99.44}, \textbf{99.55})$ \\ 
    \hline
    \multicolumn{4}{c}{$n=10,\ \epsilon=30$} \\
    \hline
    Ran-IMwG & $900.79\ (261, 1262)$ & $0.9989\ (0.9988, 0.9990)$ & $50.95\ (50.64, 51.25)$\\ 
    Sys-IMwG & $549.07\ (170, 764)$ & $0.9982\ (0.9980, 0.9983)$ & $50.94\ (50.66, 51.23)$\\ 
    SOMA & $\textbf{120.64}\ (\textbf{46}, \textbf{163})$ & $\textbf{0.9909}\ (\textbf{0.9902}, \textbf{0.9915})$ & $\textbf{91.91}\ (\textbf{91.80}, \textbf{92.01})$
  \end{tabular}
\end{table}

Table~\ref{tab:pbl_coupling1} report the coupling time of the joint posterior $(X, Y, \theta) \mid s_{\mathrm{dp}}$ as well as the convergence rate and the acceptance rate of the data imputation steps. For $\epsilon = 3$, all three methods have empirical acceptance probability higher than 94\% and they all converge fast. At $\epsilon = 30$, SOMA has an acceptance rate higher than $91\%$ while the other two IMwG samplers accept only $50\%$ of the proposed points. This translates into more rapid coupling and faster convergence of the joint $(X, Y, \theta)$ chain.

\begin{figure}[htbp]
    \centering
    \includegraphics[width=0.5\textwidth]{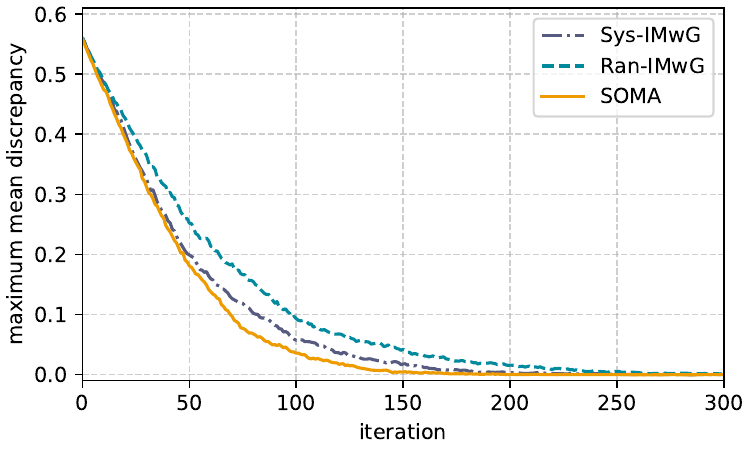}
    \caption{Evolution of the maximum mean discrepancy (MMD) between the empirical distribution of independent chains and the ground truth posterior. Lower values indicate closer agreement.}
    \label{fig:pbl_mmd}
\end{figure}

\begin{figure}[htbp]
\centering
\includegraphics[width=1.0\textwidth]{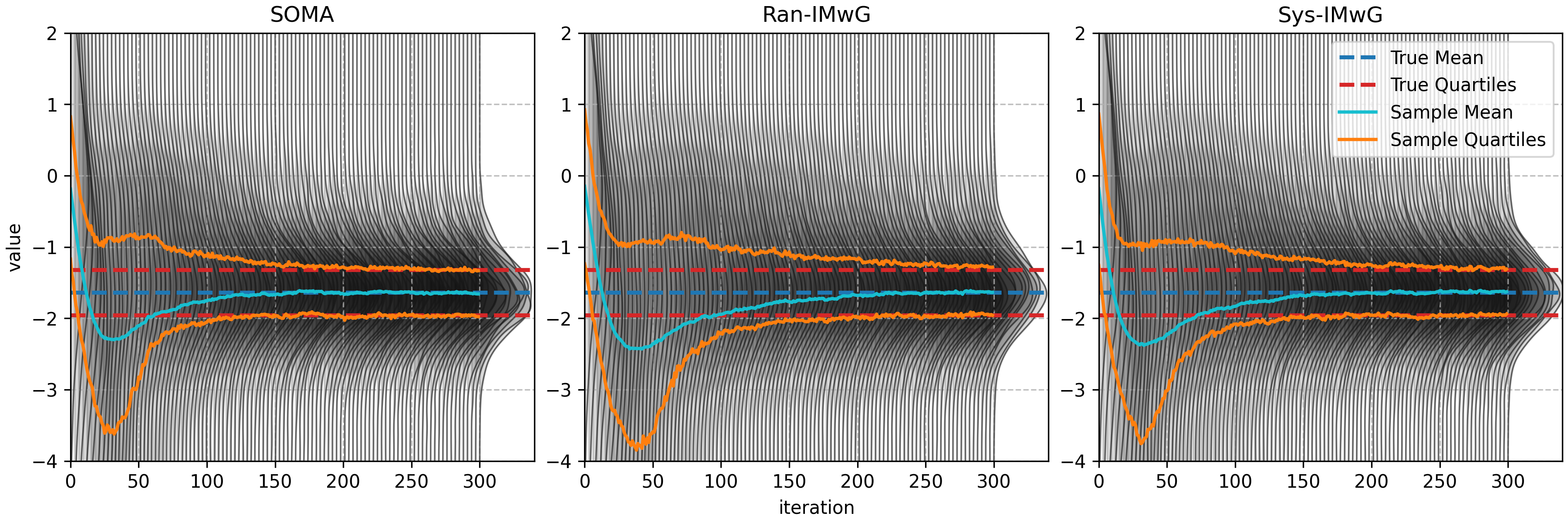}
\caption{Ridge plot of $\beta_0$ generated by running 1,000 independent single chains each for $T = 300$ iterations. The orange/cyan lines indicate the quartile/mean of $\beta_0$ across 1,000 chains. The red/blue dashed lines indicate the `true' posterior quartiles/mean, obtained from a long chain.}
\label{fig:pbl_ridge}
\end{figure}

Figure~\ref{fig:pbl_mmd} shows the Maximum Mean Discrepancy (MMD)~\citep{gretton2012kernel} between the states $\theta$ of the 1000 independent chains and the samples from the ground truth posterior distribution at each iteration under a challenging scenario where $n=100$ and $\epsilon=30$. The MMD for the SOMA decreases faster than for the other two methods, and approaches zero after 150 iterations. We give a ridge plot of the intercept term $\beta_0$ in Figure~\ref{fig:pbl_ridge}, showing how the posterior of $\beta$ given $s_{\mathrm{dp}}$ changes over MCMC iterations.
Based on visual inspections, it takes a DAMCMC algorithm using SOMA about 150 iterations to get convergence in $\beta_0$, while Ran-IMwG and Sys-IMwG require 300 and 200 iterations, respectively.

\subsection{Running Time for Different Methods}\label{sec:runtime}

Table~\ref{tab:runtime} summarizes the wall-clock times of Ran-IMwG, Sys-IMwG, and SOMA for each task, with CPU and GPU results shown separately.
With the vectorized implementation described in Section~\ref{sec:cost}, the time required by SOMA is almost less than twice that of the Ran/Sys-IMwG methods. In other words, although SOMA has to update the entire weight vector $(w_j)_{j = 0,\ldots,n}$, 
a vectorized/parallel implementation can effectively make the coefficient in front of $\mathcal{O}(n)$ small. All timings were obtained on an Intel Core i9‑10900K CPU and an Nvidia RTX 3090 GPU.

\begin{table}[htbp]
  \centering
  \tbl{Running time (seconds) for different deployment settings (mean $\pm$ sd)}
  {
  \begin{tabular}{lccc}
    Settings & Ran-IMwG & Sys-IMwG & SOMA \\
    \midrule
    \addlinespace[2pt]
    \multicolumn{4}{c}{\textbf{CPU deployment}}\\
    \addlinespace[2pt]
    \multicolumn{4}{l}{\emph{Two‑component synthetic example, 5000 iterations}}\\
    \hspace{1em}Beta prior, $\epsilon=5$   & $2.1578 \pm 0.0178$ & $1.4616 \pm 0.0140$ & $2.6208 \pm 0.0221$\\
    \hspace{1em}Beta prior, $\epsilon=10$  & $2.1750 \pm 0.0187$ & $1.4904 \pm 0.0183$ & $2.6344 \pm 0.0219$\\
    \hspace{1em}Beta prior, $\epsilon=50$  & $2.2103 \pm 0.0208$ & $1.5545 \pm 0.0186$ & $2.6941 \pm 0.0256$\\
    \hspace{1em}Exp.\ prior, $\epsilon=1$  & $1.8642 \pm 0.0218$ & $1.2175 \pm 0.0133$ & $2.3435 \pm 0.0160$\\
    \hspace{1em}Exp.\ prior, $\epsilon=2$  & $1.8794 \pm 0.0144$ & $1.2341 \pm 0.0107$ & $2.3652 \pm 0.0171$\\
    \hspace{1em}Exp.\ prior, $\epsilon=10$ & $1.9076 \pm 0.0176$ & $1.2896 \pm 0.0131$ & $2.4150 \pm 0.0228$\\
    \addlinespace[2pt]
    \multicolumn{4}{l}{\emph{Perturbed histograms, 10000 iterations}}\\
    \hspace{1em}$n=10,\ \epsilon=3$  & $1.2632 \pm 0.0213$ & $0.9367 \pm 0.0110$ & $1.6948 \pm 0.0163$\\
    \hspace{1em}$n=20,\ \epsilon=3$  & $1.2571 \pm 0.0089$ & $0.9405 \pm 0.0090$ & $1.6990 \pm 0.0160$\\
    \hspace{1em}$n=50,\ \epsilon=3$  & $1.2634 \pm 0.0111$ & $0.9457 \pm 0.0104$ & $1.7182 \pm 0.0164$\\
    \addlinespace[2pt]
    \multicolumn{4}{l}{\emph{Bayesian linear regression, 10000 iterations}}\\
    \hspace{1em}$n=100,\ \epsilon=30$ & $110.6691 \pm 0.4993$ & $102.5582 \pm 0.2933$ & $216.3009 \pm 0.4037$\\
    \hspace{1em}$n=10,\ \epsilon=30$  & $19.4360 \pm 0.0734$ & $18.4205 \pm 0.1278$ & $29.0447 \pm 0.1688$\\
    \hspace{1em}$n=10,\ \epsilon=3$   & $20.2298 \pm 0.0809$ & $18.8668 \pm 0.0845$ & $29.1283 \pm 0.1454$\\
    \midrule
    \addlinespace[2pt]
    \multicolumn{4}{c}{\textbf{GPU deployment}}\\
    \addlinespace[2pt]
    \multicolumn{4}{l}{\emph{Two‑component synthetic example, 5000 iterations}}\\
    \hspace{1em}Beta prior, $\epsilon=5$   & $8.5338 \pm 0.0562$ & $6.6322 \pm 0.0603$ & $11.3494 \pm 0.0712$\\
    \hspace{1em}Beta prior, $\epsilon=10$  & $8.5955 \pm 0.0485$ & $6.8184 \pm 0.0469$ & $11.2846 \pm 0.3150$\\
    \hspace{1em}Beta prior, $\epsilon=50$  & $8.7059 \pm 0.0683$ & $7.2848 \pm 0.0693$ & $11.2003 \pm 0.0848$\\
    \hspace{1em}Exp.\ prior, $\epsilon=1$  & $7.2237 \pm 0.1338$ & $5.6446 \pm 0.0837$ & $10.1485 \pm 0.0462$\\
    \hspace{1em}Exp.\ prior, $\epsilon=2$  & $7.3493 \pm 0.1152$ & $5.8079 \pm 0.0837$ & $10.2973 \pm 0.0523$\\
    \hspace{1em}Exp.\ prior, $\epsilon=10$  & $7.4223 \pm 0.1207$ & $6.2104 \pm 0.0721$ & $10.7829 \pm 0.0401$\\
    \addlinespace[2pt]
    \multicolumn{4}{l}{\emph{Perturbed histograms, 10000 iterations}}\\
    \hspace{1em}$n=10,\ \epsilon=3$  & $5.3686 \pm 0.1272$ & $4.4201 \pm 0.0355$ & $7.5973 \pm 0.1477$\\
    \hspace{1em}$n=20,\ \epsilon=3$  & $5.3741 \pm 0.0511$ & $4.4333 \pm 0.0452$ & $7.5172 \pm 0.0989$\\
    \hspace{1em}$n=50,\ \epsilon=3$  & $5.4059 \pm 0.0382$ & $4.4268 \pm 0.0313$ & $7.4140 \pm 0.0916$\\
    \addlinespace[2pt]
    \multicolumn{4}{l}{\emph{Bayesian linear regression, 10000 iterations}}\\
    \hspace{1em}$n=100,\ \epsilon=30$ & $534.7907 \pm 3.0974$ & $521.9027 \pm 2.3709$ & $953.7772 \pm 2.7545$\\
    \hspace{1em}$n=10,\ \epsilon=30$  & $82.4064 \pm 0.9935$ & $80.9015 \pm 1.1718$ & $125.8077 \pm 0.5780$\\
    \hspace{1em}$n=10,\ \epsilon=3$   & $84.0079 \pm 1.1297$ & $83.3762 \pm 0.8688$ & $126.1711 \pm 1.3987$\\
  \end{tabular}
  }
  \label{tab:runtime}
  \begin{tabnote}
    Running times are reported as mean $\pm$ standard deviation over 100 independent runs.  
    CPU and GPU deployments use identical algorithmic settings; differences reflect hardware acceleration only.
  \end{tabnote}
\end{table}

\end{document}